%% file: article.tex
\documentclass[a4wide]{article}

\usepackage[english]{babel} 
\usepackage[utf8]{inputenc}  %% les accents dans le fichier.tex
\usepackage[T1]{fontenc}       %% Pour la césure des mots accentués

\usepackage{amssymb,amsfonts,amsthm}

\usepackage[hscale=0.7,vscale=0.8]{geometry}
\usepackage{enumerate}
\usepackage{xspace}
\usepackage{comment}
\usepackage{graphicx,subfig}
\usepackage{multirow}
\usepackage[all]{xy}
\usepackage[hyphens]{url}
\usepackage{pgf,tikz}
\usetikzlibrary{arrows}
\usepackage{algorithm,algorithmicx}
\usepackage{algpseudocode}
\usepackage{commath}
\usepackage[toc,page]{appendix}

\newcommand{\F}{{\mathbb F}}

\newcommand{\Z}{{\mathbb Z}}

\newcommand{\N}{{\mathbb N}}
\newcommand{\A}{{\mathbb A}}

\newcommand{\OO}{{\mathcal O}}

	\DeclareMathOperator{\Gal}{Gal}

	\DeclareMathOperator{\Prob}{\mathbb{P}}

	\DeclareMathOperator{\Disc}{Disc}
	
	\DeclareMathOperator{\Res}{Res}

	\DeclareMathOperator{\Norm}{N}

	\theoremstyle{plain}
	\newtheorem{theo}{Theorem}[section]
	\newtheorem{lem}[theo]{Lemma}
	\newtheorem{prop}[theo]{Proposition}
	\newtheorem{cor}[theo]{Corollary}
	\newtheorem{fact}[theo]{Fact}

	\theoremstyle{definition}
	\newtheorem{deff}[theo]{Definition}
	\newtheorem{ex}[theo]{Example}
	\newtheorem{experiment}{Experiment}

	\theoremstyle{remark}
	\newtheorem{rem}[theo]{Remark}
	\newtheorem{nott}[theo]{Notation}

\title{Selecting polynomials for the\\ Function Field Sieve}
\date{}
\author{Razvan Barbulescu\\
	Université de Lorraine, CNRS, INRIA, France\\
\texttt{razvan.barbulescu@inria.fr}}

\begin{document}
\maketitle
\begin{abstract}
The Function Field Sieve algorithm is dedicated to computing discrete
logarithms in a finite field $\F_{q^n}$, where $q$
is a small prime power.  The scope of this
article is to select good polynomials for this algorithm by defining and
measuring the size property and the so-called root and cancellation
properties. In particular we present an algorithm for rapidly testing a
large set of polynomials. Our study also explains the behaviour of inseparable
polynomials, in particular we give an easy way to see that the algorithm
encompass the Coppersmith algorithm as a particular case.  
\end{abstract}

\section{Introduction}
%%{{{

The Function Field Sieve (FFS) algorithm is dedicated to computing discrete
logarithms in a finite field $\F_{q^n}$, where $q$
is a small prime power. Introduced by Adleman in
\cite{Ad94} and inspired by the Number Field Sieve (NFS), the algorithm collects
pairs of polynomials $(a,b)\in\F_q[t]$ such that the norms of $a-bx$ in two function fields are
both smooth (the sieving stage), i.e having only irreducible divisors of
small degree. It
then solves a sparse linear system (the linear algebra stage), whose solutions, called virtual
logarithms,  allow to compute the
discrete algorithm of any element during a final stage (individual
logarithm stage).  

The choice of the defining polynomials $f$ and $g$ for the two function fields can be seen
as a preliminary stage of the algorithm. It takes a small amount of time
but it can greatly influence the sieving stage by slightly changing the
probabilities of smoothness.
In order to solve the discrete logarithm in $\F_{q^n}$, the main required
property of $f,g\in\F_q[t][x]$ is that their resultant $\Res_x(f,g)$ has
an irreducible factor $\varphi(t)$ of degree $n$. Various methods have
been proposed to build such polynomials.

The base-$m$ method of polynomial selection, proposed by Adleman \cite{Ad94}, consists in
choosing $\varphi(t)$ an irreducible polynomial of degree $n$, setting
$g=x-m$, where $m$ is a power of $t$, and $f$ equal to the base-$m$
expansion of $\varphi$. He obtained a subexponential complexity of
$L_{q^n}(\frac{1}{3},c)^{1+o(1)}$ with $c=\sqrt[3]{64/9}$. Adleman and
Huang \cite{AdHu99} chose $\varphi$ to be a sparse polynomial and obtained a constant $c=\sqrt[3]{32/9}$. They
also noted that the previously known algorithm of Coppersmith \cite{Cop84} can be seen as a particular case of the FFS. Finally, Joux and
Lercier \cite{JoLe02} introduced a method which, without improving the
complexity, behaves better in practice. It consists in selecting a
polynomial $f$ with small degree coefficients and then of randomly testing linear polynomials $g=g_1x+g_0$ until $\Res_x(f,g)$ has an irreducible factor
of degree $n$. 
In \cite{JoLe06} Joux and Lercier proposed two additional variants of
their methods. In the first one that can be called {\em Two rational
sides}, we add the condition that $f$ has degree $1$ in $t$. Its main
advantage is that its description does not require
the theory of function fields. The second variant, called the {\em Galois
improvement}, applies to the case where $n$ is composite.

In this paper we improve the method of Joux and Lercier
\cite{JoLe02} by showing how to select the non-linear polynomial $f$.
For that we follow the strategy that was developed 
in the factorization context. In particular Murphy
\cite{Mur99} introduced and used criteria which allow
to rapidly rank any set of polynomials.  See~\cite{Bai11}, for recent developments
in this direction.

Therefore, we introduce relevant functions for the sieving efficiency of
a polynomial, taking into account a size property, a so-called root
property that reflects the behaviour modulo small irreducible
polynomials, and a cancellation property, which is analogous to the real
roots property for the NFS. We also present efficient algorithms for
quantifying these properties. A special attention is given to the
particular case where $f$ is not separable. Indeed, this is a phenomenon
that has no analogue in the factorization world and that has strong
repercussions on the sieving efficiency.

\paragraph{Recent works on composite extensions} In the past few weeks,
the algorithm in \cite{JoLe06} was the object of further improvements in
\cite{Jo12}, \cite{GGMZ13} and \cite{Jo13}, all of them being very well
adapted to the case of composite extensions. The most important of them is
Joux's new algorithm which is especially suited to the fields $\F_{q^{2k}}$
with $q$ close to $k$ and whose complexity is $L(1/4+o(1))$. Moreover, under
some overhead hidden by the $o(1)$, Joux's algorithm can be adapted to the case
$\F_{q^k}$ when $q$ and $k$ are both prime. 

In this context, the FFS keeps its interest for both theoretical and practical
reasons. On the one hand, the FFS applies to a wider set of finite fields
as the $L(1/4+o(1))$ complexity was only obtained for constant
characteristic. On the
other hand, except for the composite case, the crossing point
between the FFS and Joux's algorithm \cite{Jo13} is still to be determined.

\paragraph{Outline} The article is organized as follows.
Section~\ref{sec:functions}
lists the properties which improve the sieve yield.
Section~\ref{sec:combining} combines the previously defined functions in
order to compare arbitrary polynomials and shows how to rapidly test a large number of candidate polynomials. 
Section~\ref{insep} focuses on the case of inseparable
polynomials and, in particular, the Coppersmith algorithm.
Section~\ref{sec:examples} applies the theoretic results to a series of examples. Finally, section~\ref{conclusion} makes the synthesis.

\section{Quantification functions}\label{sec:functions}

\subsection{Size property}
%%{{{

We start by deciding the degrees in $t$ and $x$ of the two polynomials
$f$ and $g$. The FFS has always been implemented for polynomials $f$ of
small coefficients in $t$ and for polynomials $g$ of degree $1$ in $x$,
like in~\cite{JoLe02}. It might be not obvious that this is the best
choice. For instance in the case of the NFS both for factorization
\cite{Montgomery,PrZi11} and discrete logarithm
\cite{JoLe03}, pairs of non-linear polynomials were used. 
In the following, we argue that the classical choice is indeed the best
one.
\medskip
Let us first recall the nature of the objects we have to test for
smoothness. The FFS collects coprime pairs $(a(t),b(t))\in \F_q[t]$ such
that the norms of $a-bx$ in the function fields of $f$ and $g$
are both smooth. These norms are polynomials in $t$ of a simple form:
denoting $F(X,Y)=f(\frac{X}{Y})Y^{\deg_x f}$ the homogenization of
$f(x)$, the norm of $a-bx$ is just $F(a,b)$. Similarly, we denote
$G(X,Y)$ the homogenization of $g(x)$ and the second norm is $G(a,b)$.
As a consequence, the polynomial selection stage can be restated as the
search for polynomials $f$ and $g$ such that $F(a,b)$ and $G(a,b)$ are
likely to be both smooth for coprime pairs $(a,b)$ in a given range.

As a first approximation, we translate this condition into the fact that
the degree of the product $F(a,b)G(a,b)$ is as small as possible. It will
be refined all along this paper.

\begin{fact}\label{linear_g}
        Assume that we have to compute discrete logarithms in $\F_{q^n}$,
        with polynomials $f,g\in\F_q[t][x]$, such that 
	$\deg_xg\leq \deg_x f$.
        If for bounded $a$ and $b$, the polynomials $f$ and $g$ minimize 
        $\max\{\deg F(a,b)G(a,b)\}$, then we have $\deg_xg=1$.
\end{fact}
\noindent\textit{Argument.}
Let $(a,b)$ be any pair of maximal degree $e$, and assume that
there is no cancellation of the monomials in $F(a,b)$ and $G(a,b)$
respectively. Then one has
\begin{equation}
\label{norm_degree}
\deg\left(F(a,b) G(a,b)\right)=\deg_t g+e\deg_x g+\deg_t f+e\deg_xf.
\end{equation}
The degree of the resultant of $f$ and $g$ can be bounded by
$ \deg \Res(f,g) \leq \deg_xf\deg_t g+\deg_xg\deg_t f$.
Since we need this resultant to be of degree at least $n$, we impose
\begin{equation}
\label{res_degree}
\deg_xf\deg_t g+\deg_xg\deg_t f\geq n.
\end{equation}
 For a fixed value of the left hand side in
Equation~\ref{res_degree}, in order to minimize the expression in Equation
\ref{norm_degree}, we need to minimize $\deg_tf$. Therefore we set
$\deg_tf$ as small as possible, let us call this degree~$\epsilon$.
Hence the optimization problem becomes
\begin{eqnarray*}
\text{minimize} & &\left( \deg_xf+ \epsilon + e\deg_xg+\deg_t g\right)\\
	     \text{when}& &\deg_xf \deg_t g+\epsilon \deg_xg\geq n.
\end{eqnarray*}
Since one can decrease $\deg_x g$ without changing too much the left
hand side of
the constraint, the choice $\deg_x g=1$ is optimal. $\square$ 
\medskip

In the rest of the article we simply write $d$ for $\deg_xf$. The degree
of $g$ in $t$ is then about $n/d$.

\begin{rem}
    We decided to optimize the degree of the product of the norms. In
    terms of smoothness probability, it is only pertinent if both norms
    have similar degrees. More precisely, as  the
    logarithm of Dickman's
    rho function is concave, it can be shown that it is optimal to balance the
    degrees of the norms. Hence sensible choices of the parameters are
    such that
$de\approx \frac{n}{d}+e$.
\end{rem}

%%}}}
%%{{{

We are now ready to quantify the size property for a single polynomial
$f$. It clearly depends on the bound $e$ on both $\deg a$ and $\deg b$. 
In the following definition, we also take into account the skewness
improvement as implemented in \cite{DeGaVi12}, that is, we set the
skewness $s$ to be the difference between the bounds on the degree of $a$
and of $b$.  We can then define sigma to match the
average degree of $F(a,b)$ for $(a,b)$ in a domain of polynomials of degrees
bounded by $\lfloor e+s/2\rfloor$ and $\lfloor e-s/2\rfloor$, when no cancellation occurs among the
monomials of $F(a,b)$. This translates into the following formal
definition.

\begin{deff}  Let $f\in\F_q[t][x]$ be a polynomial, $s$ the skewness parameter 
	and $e$ the sieve size parameter. We define:
	\begin{equation}
		\sigma(f,s,e)=\sum_{\begin{array}{c}0\leq d_a\leq
			e+s/2\\0\leq d_b \leq
		e-s/2\end{array}}p_{d_a,d_b}\max_{i\in[0,d]}\bigg(\deg(f_i)+id_a+(d-i)d_b\bigg),
	\end{equation}
	with $p_{d_a,d_b}=(q-1)^2q^{d_a+d_b}/q^{\lfloor
e+s/2\rfloor+\lfloor e-s/2\rfloor+2}$.
\end{deff}
%%}}}
%%{{{
\subsection{Root property}

As a first approximation, for a random pair $(a,b)\in \F_q[t]^2$,
$F(a,b)$ has a smoothness probability of the same order of magnitude as
random polynomials of the same degree. Nevertheless, we shall show that for a
fixed size property, some polynomials improve this probability by a
factor of $2$ or more. 

Consider the example of
$f=x(x-1)-(t^{64}-t)$ in $\F_2[t][x]$. For all monic irreducible
polynomials $\ell$ of
degree at most $3$, $\ell$ divides the constant term, so $f$ has two roots
modulo $\ell$. For each such $\ell$ and for all $b$ non divisible by $\ell$,
there are $2$ residues of $a$ modulo $\ell$ such that $F(a,b)\equiv 0\mod
\ell$. Therefore, for all $\ell$ of degree $3$ or less, the probability
that $\ell$ divides the norm is heuristically twice as large as the
probability that it divides a random polynomial. This influences in turn the
smoothness probability. This effect is quantified by the function alpha
that we now introduce.

\subsubsection{Definition of alpha}

Introduced by Murphy \cite{Mur99} in the case of the NFS, alpha can be
extended to the case of the FFS.  Let $\ell$ be a monic irreducible
polynomial in $\F_q[t]$ and let us call $\ell$-part of a polynomial $P$,
the largest power of $\ell$ in $P$.  We shall prove that the quantity
$\alpha_\ell$ below is the degree difference between the $\ell$-part of a
random polynomial and the $\ell$-part of $F(a,b)$ for a
random coprime pair $(a,b)\in\F_q[t]$.

Let us first properly define the average of a function on a set
of polynomials.

\begin{deff}
Let $v$ be a real function of one or two polynomial variables
$v:\F_q[t]\rightarrow \mathbb{R}$ or $v:\F_q[t]\times \F_q[t]\rightarrow
\mathbb{R}$. Let $S$ be a subset of the domain of $v$.
For any pair $(a,b)$ of polynomials, we write
$\deg(a,b)$ for $\max(\deg(a),\deg(b))$.
If the limit below exists we call it
average of $v$ over $S$ and denote
it by $\mathbb{A}(v,S)$:
\begin{equation*}
    \mathbb{A}(v,S) = 
\lim_{N\rightarrow\infty}\frac{\sum_{s\in S,\deg(s)\leq
N}\varphi(s)}{\#\{s\in S\mid \deg(s)\leq N\}}.  
\end{equation*} 
\end{deff}

\begin{deff}
	Let $L$ denote the set of monic irreducible polynomials in
	$\F_q[t]$. Put $D=\{(a,b)\in \F_q[t]^2\mid
\gcd(a,b)=1\}$. Take a non-constant polynomial $f\in \F_q[t][x]$. When the right hand
	members are defined, we set for all $\ell\in L$:
\begin{eqnarray*}
\alpha_\ell(f)&=&\deg(\ell)\ \Big(\A\left(v_\ell(P),\{P\in\F_q[t]\}\right)
    -\A\left(v_\ell(F(a,b)),\{(a,b)\in D\}\right)\Big),\\
	\alpha(f)&=&\sum_{\ell \in L}\alpha_\ell(f),
\end{eqnarray*} 
where $v_\ell$ is the valuation at $\ell$. The infinite sum which defines
$\alpha(f)$ must be seen as a formal notation and by its sum we denote the
limit when $b_0$ goes to infinity of $\alpha(f,b_0):=\sum_{\ell\in L,\deg
\ell\leq b_0}\alpha_\ell(f)$.
\end{deff}

\begin{nott}
	For all irreducible polynomial $\ell\in\F_q[t]$, $\Norm(\ell)$
denotes the number of residues modulo $\ell$, i.e. $q^{\deg \ell}$.
\end{nott}

We call affine root of $f$ modulo $\ell^k$ any $r\in \F_q[t]$ such that
$\deg r< k \deg \ell$ and $F(r,1)=f(r)\equiv 0\mod \ell^k$. Also we call
projective root of $f$ modulo $\ell^k$ the polynomials $r\in\F_q[t]$ such
that $\ell \mid r$, $\deg r< k \deg\ell$ and $F(1,r)\equiv 0\mod \ell^k$.
Note that, when $k=1$ any affine and projective roots can be seen as
an element of $\mathbb{P}^1(\F_q(t))$, hence we denote them $(r:1)$ and
$(1:r)$ respectively. 

\begin{nott} We denote by $S(f,\ell)$ the set of affine and projective
    roots modulo $\ell^k$ for any $k$:
	\begin{equation}\label{S}
		S(f,\ell)=\left\{(r,k)\mid k\geq 1, r\text{ affine or
			projective root of $f$ modulo }\ell^k \right\}.
\end{equation}
\end{nott}

\begin{prop} \label{ca} Let $f\in\F_q[t][x]$ and $\ell\in\F_q[t]$ a monic irreducible
	polynomial. Then $\alpha_\ell$ exists and we have 
	\begin{equation}\label{approx}
		\alpha_\ell(f)=\deg \ell
		\left(\frac{1}{\Norm(\ell)-1}-\frac{\Norm(\ell)}{\Norm(\ell)+1}\sum_{(r,k)\in
		S(f,\ell)}\frac{1}{\Norm(\ell)^k}\right).
	\end{equation}
\end{prop}
\begin{proof}

    In order to prove the convergence of Equation~\ref{approx}, note that some elements of
	$S(f,\ell)$ group into infinite sequences
	$\{(r^{(k)},k)\}_k$ with $r^{(k)}\equiv
r^{(k-1)}\mod \ell^{k-1}$. Since each infinite sequence defines a root of
$f$ in the $\ell$-adic completion of
$\F_q(t)$, there are at most $d$ such sequences, whose contributions
converge geometrically. There are only finitely
many remaining elements of $S(f,\ell)$ because otherwise one could extract an
additional $\ell$-adic root. This proves the convergence. 

In order to show the equality, note that we have
$\A(v_\ell(P),\F_q[t])=\sum_{k=1}^\infty\frac{1}{\Norm(\ell)^k}=\frac{1}{\Norm(\ell)-1} $.
Indeed, for all $k\in\N$,
the density of the set of polynomials divisible by $\ell^k$ is the inverse
of the number of residues modulo $\ell^k$, which is $q^{\deg(\ell)k}=\Norm(\ell)^k$.

Let us compute $a_{hom}^{(\ell)}:=\A(v_\ell(F(a,b)),\{(a,b)\in \F_q[t]^2\mid
\gcd(a,b)\not\equiv 0 \bmod \ell\})$. This corresponds to the
contribution of $F(a,b)$ in the definition of $\alpha_\ell$. The
condition $\gcd(a,b)=1$ has been replaced by a local condition. Since we
are only interested in $\ell$-valuations, this does not change the
result.
The number of $\ell$-coprime pairs $(a,b)$ of degree less than
$k\deg \ell$ is $\Norm(\ell)^{2k}-\Norm(\ell)^{2k-2}$. Such a pair $(a,b)$
satisfies $\ell^k\mid F(a,b)$ if and only if $\ell\nmid b$ and $ab^{-1}$
is an affine root modulo $\ell^k$ or $\ell\mid b$ and $ba^{-1}$ is a
projective root modulo $\ell^k$.
For each affine root
$r$, the number of $\ell$-coprime pairs $(a,b)$ such that $a\equiv br\mod \ell^k$
equals the number of choices for $b$, which is
$\Norm(\ell)^k-\Norm(\ell)^{k-1}$. Similarly, for each projective root, the
number of coprime pairs $(a,b)$
such that $b\equiv ar\mod \ell^k$ is the number of choices for $a$, which is
$\Norm(\ell)^k-\Norm(\ell)^{k-1}$. Hence, it follows that
for each $(r,k)\in S(f,\ell)$,
the probability that the residues of a coprime pair $(a,b)$ modulo
$\ell^k$ match the root $(r,k)$ is given by the formula below, 
where $(a:b)\equiv r \mod \ell^k$ is short for
$r\equiv ab^{-1}\mod \ell^k$ when $\ell\nmid b$ or $r\equiv ba^{-1}\mod
\ell^k $ when $\ell\mid b$. 
\begin{equation}\label{P(r)}
\Prob((a:b)\equiv r \bmod \ell^k)=
\frac{1}{\Norm(\ell)^k}\frac{\Norm(\ell)}{\Norm(\ell)+1},
\end{equation}

The following equation is intuitive so that we postpone its proof until
Lemma~\ref{lem_appendix} in the Appendix. The technicalities arise due to
the fact that calling the quantities ``probability'' is formally
incorrect, and must be replaced by natural densities.
\begin{equation}\label{appendix}
	\begin{array}{rcl}
		a_{hom}^{(\ell)}(f) &=&\sum_{k=1}^\infty k\Prob(v_\ell(F(a,b))=k)\\
	=\sum_{k=1}^\infty\Prob(v_\ell(F(a,b))\geq k)&=&\sum_{(r,k)\in
	S(f,\ell)}\Prob((a:b)\equiv r \bmod \ell^k).
\end{array}
\end{equation}
Replacing $a_{hom}^{(\ell)}$ in formula $\alpha_\ell(f)=\frac{\deg
\ell}{\Norm(\ell)-1}-\deg(\ell)a_{hom}^{(\ell)}(f)$ and using Equation~\ref{P(r)} yields the desired result.

\end{proof}
If $f$ has only simple roots modulo $\ell$, then by Hensel's Lemma $f$ has
the same number of roots modulo every power of $\ell$. We obtain the
following formula.
\begin{cor}\label{cor_ell}
    Let $f\in\F_q[x][t]$ and $\ell$ a monic irreducible polynomial
	in $\F_q[t]$ such that the affine and projective roots of $f$
        modulo $\ell$ are simple and call $n_\ell$ their number. Then 
	\begin{equation}\label{cor_ell_eq}
		\alpha_\ell(f)=	\frac{\deg
		\ell}{\Norm(\ell)-1}\left(1-\frac{
                \Norm(\ell)}{\Norm(\ell)+1} n_\ell\right).
	\end{equation}
\end{cor}

\subsubsection{The case of linear polynomials}
%%{{{
Showing that $\alpha(f)$ converges is not trivial and, to our knowledge, it
is not proven in the NFS case. Let us first show that
$\alpha(g)$ converges for linear polynomials $g$.
This requires the following classical identity.

\begin{lem} (Chapter I,\cite{Ap90}) \label{Möbius}
	Let $\mu$ denote Möbius' function and let $x$ be such that $|x|<1$.
        Then we have
	\begin{equation*}
	\sum_{h\geq
	1}\mu(h)\frac{x^h}{1-x^h}=x.
\end{equation*}
\end{lem}
\iffalse
\begin{proof}
	The double series $\sum_{h\geq 1}\sum_{n\geq 1}\mu(h)x^{nh}$
	converges absolutely. Indeed, $Im( \mu)=\{-1,0,1\}$ and we have $\sum_{h\geq
	1}\sum_{n\geq 1} x^{hm}=\sum_{h\geq 1}\frac{x^h}{1-x^h}\leq
	\sum_{h\geq 1}2 x^h=\frac{2x}{1-x}\leq2$. Therefore we can change
	the summation order:
	\begin{equation*}
		\sum_{h\geq 1}\sum_{n\geq 1}\mu(h)x^{hn}=\sum_{N\geq
		1}\sum_{h\mid N}\mu(h)x^N.
	\end{equation*}
	Or the Möbius function satisfies:
	\begin{equation*}
		\sum_{h\mid N}\mu(h)=\left\{\begin{array}{l}1,\text{ if }N=1\\
				                      0, \text{ otherwise}
			             \end{array}\right.
\end{equation*}
This completes the proof.
\end{proof}
\fi
\begin{nott}
	We denote by $I_k$ the number of monic degree-$k$ irreducible
        polynomials in $\F_q[t]$.
\end{nott}
\begin{theo}\label{alpha_x}
	Let $g\in\F_q[t][x]$ be such that $\deg_xg=1$. Then 
	\begin{equation*}
		\alpha(g)=\frac{1}{q-1}.		
	\end{equation*}
\end{theo}
\begin{proof}
	Let $L$ be the set of monic irreducible polynomials in $\F_q[t]$.
        We shall prove that, when $b_0$ goes to infinity, $\sum_{\ell\in L,\deg\ell\leq b_0}
           \alpha_\ell(g)$ tends to $\frac{1}{q-1}$.
	In Equation~\ref{cor_ell_eq} one has $n_\ell=1$ and therefore 
	\begin{equation}
	\sum_{\ell\in L,\deg\ell\leq b_0}\alpha_\ell(g)=
          \sum_{\ell \in L, \deg \ell \leq b_0}
		\frac{\deg(\ell)}{\Norm(\ell)^2-1}=
          \sum_{k\leq b_0}\frac{k I_k}{q^{2k}-1}.
\end{equation}

Since $kI_k=\sum_{h\mid k}\mu(h)q^\frac{k}{h}$, the series transforms into a
double series for which we shall prove the absolute convergence and shall
compute the sum: 
\begin{equation*}
\sum_{k\geq 1}\sum_{h\mid
k}\frac{1}{q^{2k}-1}\mu(h)q^\frac{k}{h}.
\end{equation*}
The absolute value of the term
$\left|\mu(h)\frac{q^{\frac{k}{h}}}{q^{2k}-1}\right|$
is bounded by 
$\frac{q^{\frac{k}{h}}}{q^{2k}-1}$. It follows easily that the sum is
bounded by $\frac{4}{q-1}$. 
%Since $q^{2k}-1\geq \frac{1}{2} q^{2k}$, we can further bound it by $2\sum_{k\geq 1,h\mid k}\frac{q^{\frac{k}{h}}}{q^{2k}}=2\sum_{h\geq 1,i\geq 1}\frac{q^{i}}{q^{2hi}}=2\sum_{i\geq1}\frac{q^i}{q^{2i}-1}\leq 4 \sum_{i\geq 1}\frac{1}{q^i}=4\frac{1}{q-1}<\infty$. 
Therefore, we can change the
summation order: 
\begin{equation}
	\label{moeb}
	\sum_{k\geq 1}\sum_{h\mid
	k}\mu(h)\frac{q^{\frac{k}{h}}}{q^{2k}-1 }=\sum_{i\geq 1}\left(\sum_{h\geq 1}\mu(h)\frac{q^i}{q^{2hi}-1}\right).
\end{equation}

Applying Lemma~\ref{Möbius} to $x=\frac{1}{q^{2i}}$ leads to $\sum_{h\geq
1}\mu(h)\frac{1}{q^{2hi}-1}=\frac{1}{q^{2i}}$. This shows that, when $b_0$
goes to infinity, $\sum_{\deg\ell\leq b_0}\alpha_\ell(g)$ tends to $\sum_{i\geq
1}\frac{1}{q^i}=\frac{1}{q-1}$. 

\end{proof}

%%}}}
%%{{{
We conclude the subsection by showing the convergence of alpha. Let now
$\mathcal{C}$ be a (singular or non-singular) projective curve. Call $\mathcal{P}_k(\mathcal{C})$ the set of points of
$\mathcal{C}$ with coefficients in $\F_{q^k}$ and $P_k(\mathcal{C})$ its
cardinality. Next, call $\mathcal{P}_k'(\mathcal{C})$ 
the set of points in $\mathcal{P}_k(\mathcal{C})$ whose $t$-coordinate does not
belong to a strict subfield of $\F_{q^k}$. Finally, we denote by 
$P_k'(\mathcal{C})$ its cardinality.

We shall need the following intermediate result.
\begin{lem}\label{P'P} Let $\mathcal{C}$ be projective plane curve of degree
	$d_0$ defined over $\F_q$ 
        and let $g_0=(d_0-1)(d_0-2)/2$ be its arithmetic genus.
	Then, for all $k\geq 1$, 
	\begin{equation}\label{lemme}
		\left|P_k'(\mathcal{C})-(q^k+1)\right|<(4g_0+6)q^{\frac{k}{2}}.
	\end{equation}
\end{lem}
\begin{proof}
    Let $\tilde{\mathcal{C}}$ be a non-singular model for
        $\mathcal{C}$ and let $g$ be its geometric genus.
	We apply the Hasse-Weil Theorem and obtain:
\begin{equation}\label{Hasse-Weil}
	\left| P_k(\tilde{\mathcal{C}})-(q^k+1)\right|\leq 2g\cdot
	q^{\frac{k}{2}}.
\end{equation}
Next, according to Chapter VI in \cite{Ful69},
\begin{equation}\label{desingular}
	|P_k(\tilde{\mathcal{C}})-P_k(\mathcal{C})|\leq g_0-g.
\end{equation}
Every point $(t_0,x_0)$ of $\mathcal{P}_k(\mathcal{C})\backslash
\mathcal{P}_k'(\mathcal{C})$ is determined by the choice of: a) a strict
divisor $d$ of $k$, b) an element $t_0\in \F_{q^{\frac{k}{d}}}$ and c) a root
$x_0$ of $f(t_0,x)$ in $\F_{q^k}$. Therefore:
\begin{equation}\label{raw_lemma}
	\left| P_k'(\mathcal{C})-P_k(\mathcal{C})\right|\leq \sum_{d\mid
	k,1<d\leq k}q^{\frac{k}{d}}\deg_x(f).
\end{equation}
On the one hand $\deg_xf< g_0+3$; on the other hand, if one calls $d_1$ the
smallest proper divisor of $k$, $\sum_{d\mid
k,1<d\leq k}q^{\frac{k}{d}} = q^{\frac{k}{d_1}} \sum_{d\mid
k,1<d\leq k}q^{\frac{k}{d}-\frac{k}{d_1}}$. Since $d\mid k$ and $d\geq d_1$,
$\frac{k}{d}-\frac{k}{d_1}$ is a negative or null integer. Therefore this sum is
bounded by $q^{\frac{k}{d_1}}\sum_{i\geq 0}q^{-i}\leq q^{\frac{k}{d_1}}\sum_{i\geq
0}2^{-i}=2q^{\frac{k}{d_1}}$. But $d_1\geq 2$, so 
\begin{equation}\label{P'}
	\left|P_k'(\mathcal{C})-P_k(\mathcal{C})\right|\leq 2
	q^{\frac{k}{2}}\deg_x f < 2(g_0+3)q^{\frac{k}{2}}.
\end{equation}
The result follows from Equations~\ref{Hasse-Weil}, \ref{desingular} and \ref{P'}.
\end{proof}
The following theorem shows that $\alpha(f)$ is well defined. More
specifically, call $L$
the set of irreducible monic polynomials in $\F_q[t]$.  We show the existence of alpha
by showing the convergence of $v_{b_0}:=\sum_{\ell\in L,\deg
\ell\leq b_0}(\alpha_\ell(f)-\alpha_\ell(x))$.

\begin{theo} \label{alpha_convergence}Let $f\in\F_q[t][x]$ an absolutely
	irreducible separable polynomial.  Then the sequence $v_{b_0}$ defined above converges. If one defines alpha
	by $\alpha(f)=\lim_{b_0\rightarrow\infty} \sum_{\deg \ell\leq b_0, \ell\in L}\alpha_\ell(f)$, then there exist explicit bounds on $\alpha$ that depend only on $\deg_xf$, $\deg_t f$ and $q$. 

\end{theo}
\begin{proof}
Call $d_0=\deg f$ the degree of $f$ as a polynomial in two variables and
$g_0=(d_0-1)(d_0-2)/2$. Let $L_0$ be the set of
irreducible divisors of $\Disc(f)\cdot f_d$. Call $b_0$ the largest degree of elements in $L_0$. Let
        $k > b_0$. We are in the case of Corollary~\ref{cor_ell}, hence
        Equation~\ref{cor_ell_eq} gives
\begin{eqnarray*}\label{eq1}
	\sum_{\ell\in L, \deg \ell=k}
	\alpha_\ell(f)-\alpha_\ell(x)=\qquad\qquad\qquad\qquad\qquad\\
\qquad=\frac{kq^k}{q^{2k}-1}\left(\#\{(\ell,r)\mid
	\deg(\ell)=k, f(r)\equiv
0\mod\ell\}-I_k\right).
\end{eqnarray*}
Each pair $(\ell,r)$ as in the equation above corresponds to exactly $k$ points on
the curve $\mathcal{C}$ associated to $f$. Indeed, each $\ell$ has
exactly $k$ distinct roots in $\F_{q^k}$.
Hence we have $I_k=\frac{1}{k}P_k'(\mathbb{P}^1(\F_q))$
and $\#\{(\ell,r)\mid f(r)\equiv
0\mod\ell\}=\frac{1}{k}P_k'(\mathcal{C})$, and further:
\begin{equation}\label{eq2}
\left |\#\{(\ell,r)\mid f(r)\equiv
0\mod\ell\}-I_k\right|=\frac{1}{k}\left|P_k'(\mathcal{C})-P_k'(
\mathbb{P}^1(\F_q)) \right|.
\end{equation}
Finally, Lemma~\ref{P'P} applied to $\mathcal{C}$ and
$\mathbb{P}^1(\F_q))$
respectively gives
\begin{eqnarray}\label{double}
	\left|P_k'(\mathcal{C})-(q^k+1)\right|&\leq &(4g_0+6)\sqrt{q^k}\\
	      \left|P_k'(\mathbb{P}^1(\F_q))-(q^k+1)\right|&\leq
                                                    &6\sqrt{q^k},
\end{eqnarray}
where $g_0$ is the arithmetic genus of $\mathcal{C}$.

Hence
$\left| \#\{(\ell,r)\mid f(r)\equiv
0\mod\ell\}-I_k \right|\frac{kq^k}{q^{2k}-1}\leq
(4g_0+12)\frac{q^k}{q^{2k}-1}\sqrt{q^k}$. The series $\sum_{k\geq
1}\frac{q^k}{q^{2k}-1}\sqrt{q^k}$ is equivalent to the series $\sum_k
\sqrt{q}^{-k}$ which converges. Therefore the sequence
$v_{b_0}$ converges when $b_0$ tends to infinity.

For a given pair $d_0$ and $q$, one can clearly bound the set $L_0$.
For all $\ell\in L_0$, by Proposition~\ref{ca}, $S(f,\ell)$ is formed by at most
$\deg_xf\leq d_0$ infinite sequences and a finite number of additional
elements. We are thus left with finding a bound for the roots which do not
extend into $\ell$-adic roots. By Hensel's
Lemma, if a root $(r,k)$ does not lift to an $\ell$-adic one, we have
$f'(r)\equiv 0\mod \ell^k$. This implies $\Disc(f)\equiv 0\mod \ell^k$ which
gives a bound on $k$. Therefore, alpha admits an effective bound depending exclusively on $q$ and $d_0$. 
\end{proof}
\begin{ex}
     Following the proof of the previous theorem, we can not only find a
     bound on $\alpha$, but also evaluate the speed of convergence.
	Take $q=2$, and let $f\in \F_q[t][x]$ such that $\deg_x f=6$
	and $\tilde{g}=19$ and suppose that $L_0$ contains only polynomials of degree
	less than $15$. Using Equations~\ref{eq2} and \ref{double} in the
	proof above and the exact formula for $I_k$ we can prove that
	$\alpha(f)$ is computed up to an error of $0.567$ if we sum
	polynomials $\ell$ up to degree $15$ and we reduce the error to
	$0.097$ if we go to degree $20$. 
\end{ex}

%%}}}
%%}}}
\subsection{Cancelation property-Laurent roots}
%%{{{
Consider the polynomial $f=x^3+t^2x+1\in\F_2[t][x]$. For all $(a,b)\in
\F_2[t]^2$, if no cancellation occur, the degree of $F(a,b)=a^3+t^2ab^2+b^3$ is $\max(\deg
a^3,\deg(t^2ab^2),\deg b^3)$. One can easily check that the degree of
$F(a,b)$ is lower than this value if and only if $\deg a-\deg b$
equals $1$ or $-2$. Moreover, in the first case the decrease is at least $2$
while in the second case it is at least $1$. These conditions can be better
explained thanks to the Laurent series.

We call Laurent series (in $1/t$) over $\F_q$ any series 
$\sum_{n\geq n_0} a_n \frac{1}{t^n}$ with $n_0\in\Z$ and coefficients $a_n$
in $\F_q$. We make the common convention to call degree of a rational fraction $f_1/f_2$ with
$f_1,f_2\in\F_q[t]$ the difference $\deg f_1-\deg f_2$. The degree of a 
Laurent series is then defined as the degree of any of its nonzero truncations e.g
$t+1+1/t^2+\cdots$ has degree $1$. Equivalently, it is the opposite of the
valuation of the Laurent series in $1/t$. We call Laurent polynomial a pair $(r,m)$
such that $r\in \F_q(t)$, $m$ is an integer and $r$ is the sum of a Laurent series whose terms
are null starting from index $m+1$. We may also use
$r+O(\frac{1}{t^{m+1}})$ for
writing the Laurent polynomial $(r,m)$. 

Formally, a pair $(a,b)$ has a ``decrease in degree'' if $\max_i \deg( f_i
a^i b^{d-i})$ is strictly larger than $\deg F(a,b)$. Note that $F(a,b)$
has a decrease in degree if and only if $b\neq 0$ and the first terms of the Laurent series
$\frac{a}{b}$ match those of a Laurent polynomial $r$ with the property
given in the following definition.

\begin{deff}\label{def-Laurent}
	Let $f\in\F_q[t][x]$ be a polynomial and call $d$ its degree in $x$.
	Let $(r,m)$ be a Laurent polynomial. We say that $(r,m)$ is a Laurent root of $f$ if \begin{equation}\label{r}
	\max_{i\in[0,d]}\deg(f_ir^i)-\deg f(r) > 0.
	\end{equation}
	We call gap of $(r,m)$ the least value in the left hand side of the
	inequality above when we
	replace $r$ by any Laurent series extending $r$. A Laurent series such that
all its truncations are Laurent roots is called an infinite Laurent root. 
\end{deff}
In the example above, $\frac{1}{t^2}+\frac{1}{t^8}+O(\frac{1}{t^9})$ is a
Laurent roots of gap $7$ and it extends into an infinite Laurent root. Also
$t+O(1)$ is a Laurent root of gap $2$ that is not the truncation of any
Laurent root with a larger $m$. It also shows that the gap is not directly
connected to the number of terms in the Laurent polynomial.

\subsubsection{Computation} 
One can compute every Laurent root in two steps. First, one computes the Laurent
roots of type $\lambda t^\delta$ with $\lambda$ is in $\F_q$ and $\delta$
is an integer. For this, call
Newton polygon of $f$, with respect to valuation $-\deg$, the convex hull of
$\{(d-i,\deg(f_{i}))\mid i\in[0,d]\}$. Chapter II in \cite{Neu99} shows that
$\delta$ must be an integer slope of the Newton
polygon of $f$. Next, to extend a Laurent root 
$r=a_{n_0}t^{n_0}+\cdots+a_m\frac{1}{t^m}$ with $a_m\neq 0$ to a root
with precision larger than $m$, one computes the Laurent
roots $\lambda t^\delta$ of $f(x+r)$ for which $\delta$ is an integer such that
$\delta < -m$. Note that this corresponds to make a Hensel lift with respect to
the valuation $-\deg$.

In order to compute the gap of a Laurent root $(r,m)$, note that in Equation~\ref{r} the term $\max_{i\in[0,d]}\deg(f_ir^i)$ depends
only on the leading term of $r$. Hence the problem is reduced to that of
computing the maximal degree of
$f(R)$ for the Laurent series $R$ which extend $r$.  For this, one sets an
upper bound and then tests Laurent
polynomials with increasingly more terms and reduces the upper bound
until they produce a certificate. 

%%}}}
%%{{{
\subsubsection{Definition of $\alpha_\infty$.}\label{sssec:alpha_infty}
For each Laurent root $(r,m)$, we can compute the proportion of pairs $(a,b)$ on a
sieve domain such that the first terms of $a/b$ match $(r,m)$. 
 Recall that a sieve domain of sieve parameter $e$ and skewness $s$ corresponds
 to all the pairs $(a,b)$ with $\deg a\leq\lfloor e+s/2\rfloor$ and $\deg b\leq
 \lfloor e-s/2 \rfloor$. 
 
 \begin{lem}\label{prob_r}
	Let $r+O(\frac{1}{t^{m+1}})$ be a Laurent root of $f\in\F_q[t][x]$. 
 Call $N_r=\deg(r)+m$, the number of terms of $r$ other than
the leading one.
Then the proportion of pairs $(a,b)$ on a domain of sieve parameter $e\geq
N_r+|\deg r|$ and
skewness parameter $s$ such that the Laurent series $a/b$
matches $r+O(\frac{1}{t^{m+1}})$ is:
\begin{equation*}
	\Prob
        \left(\frac{a}{b}=r+O\left(\frac{1}{t^{m+1}}\right)\right)
        =\frac {q^{-N_r-|s-\deg r|}} {q+1}
        (1+O_{e\rightarrow\infty}(1/q^{2e})).
\end{equation*}
\end{lem}
\begin{proof}
 The proportion of pairs $(a,b)$ such that $\deg a-\deg b=\deg r$ is
 approximated by $(\frac{q-1}{q})^2\cdot$ $\cdot\sum_{i\geq
0}\frac{1}{q^{2i}} q^{-|s-\deg r|}=$$\frac{q-1}{q+1}q^{-|s-\deg r|}$. See
Figure~\ref{skewness} for an illustration. The relative error made is
$O(1/q^{2e})$ and corresponds to the fact that the series
$\sum\frac{1}{q^{2i}}$ must be truncated at $i= e-|\deg r|$ and to the fact
that for $\deg a,\deg b< |s|$ there is no pair $(a,b)$ such that $\deg
a-\deg b=s$. Next, only a fraction $\frac{1}{q-1}$
of these pairs have leading coefficients such that $a/b=r(1+O(\frac{1}{t}))$. Finally, when
$a/b=r(1+O(\frac{1}{t}))$, the condition
$a/b=r+O(\frac{1}{t^{m+1}})=r(1+O(\frac{1}{t^{N_r+1}}))$ can be expressed as a
system of $N_r$ linear equations, that is triangular on the variables of $a$.
Therefore, a pair with $\deg a-\deg b=\deg r$ and $a/b=r(1+O(1/t))$ has
probability $q^{-N_r}$ to satisfy $a/b=r(1+O(\frac{1}{t^{N_r+1}}))$. This
leads to the proportion announced in the statement. Note that the condition
$e> N_r+|\deg r|$ guarantees that the polynomials $a$ and $b$ have
sufficiently many coefficients so that the linear conditions make sense.
\end{proof}
We can now define a function which, for large sieve domains, measures the
average degree gained due to the cancellations.

\begin{deff}
    For any Laurent root $r+O(\frac{1}{t^{m+1}})$ we call
    $\mathrm{trunc}(r)$ the
	Laurent polynomial obtained by deleting the term in $\frac{1}{t^m}$.
        If $\mathrm{trunc}(r)\neq 0$ we write $\gamma(r,m)$ for the gap of
        $r+O(\frac{1}{t^{m+1}})$ minus the gap of $\mathrm{trunc}(r)+O(\frac{1}{t^m})$.
	Otherwise $\gamma(r,m)$ is the gap of $r+O(\frac{1}{t^{m+1}})$. We
	call alpha infinity the following quantity
	\begin{equation}\label{alpha_infty}
		\alpha_\infty(f,s):=\sum_{(r,m)\text{ Laurent
		root}} -\gamma(r,m)\frac{1}{q+1} q^{-N_r-|s-\deg(r)|}.
\end{equation}
\end{deff}

Consider the case when all the Laurent roots of $f$ extend infinitely into the
Laurent series $r_1$, $r_2$, $\ldots$, $r_h$. If, for all $i$, each new term
of $a/b$ which matches $r_i$ increases the gap by one, then we obtain the simpler formula below.

	\begin{equation}
		\alpha_\infty(f,s)=\frac{-q}{q^2-1}\sum_{i=1}^h q^{-|s-\deg r_i|}.
	\end{equation}
As a particular case, it is clear that the polynomials of degree $1$
 in $x$ have exactly one infinite Laurent root. If one sets the skewness
 $s=\deg_t f$, then $\alpha_\infty(f,s)=-\frac{q}{q^2-1}$. As a second
 example, a degree-$6$ polynomial $f$ over $\F_2$ which has infinite Laurent
 roots of degrees $3,2,1,0,-1$ and $-2$ has
 $\alpha_\infty(f,0)=-1.75$. 
 \begin{ex} 
	 A special class of polynomials are those corresponding to
	 $\mathcal{C}_{\text{a},\text{b}}$ curves, which were proposed
         for the FFS in~\cite{Mat99}. If $f$ is a
	 $\mathcal{C}_{\text{a},\text{b}}$ polynomial and if we denote
	 $\text{a}=\deg_x f$ and
	 $\text{b}=\deg_t f_0-\deg f_\text{a}$, then 
	 for all $i\in[0,\text{a}-1]$ we have $\deg f_i<
	 \deg f_\text{a} +(\text{a}-i)\frac{\text{b}}{\text{a}}$. 
	 
	 Suppose that a $\mathcal{C}_{\text{a},\text{b}}$ polynomial $f$ had a Laurent
	 root $r$. If $\deg r<\frac{\text{b}}{\text{a}}$, then $\max \{\deg f_i r^i\mid
	 i\in[1,\text{a}]\}< \deg f_0$. If
	 $\deg r\geq \frac{\text{a}}{\text{b}}$ then $\deg
	 f_\text{a} r^\text{a}$ dominates all the
	 other terms of $f(r)$, so $f$ has no Laurent roots. Hence, for any
	 $s$, $\alpha_\infty(f,s)=0$. 
\end{ex}
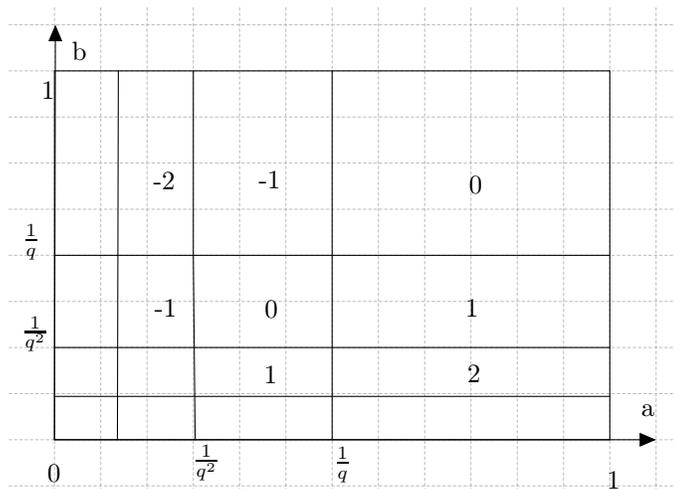
\begin{figure}[ht!]
\caption{A domain of pairs $(a,b)$ in lexicographical order having skewness
$S$. We write the quantity $\deg a -\deg b -S$ for each region.}
\label{skewness}
\begin{center}
\input{skewness.1.tex}
\end{center}
\end{figure}
%%}}}
%%{{{
\subsubsection{Constructing polynomials with many Laurent roots}\label{many}
One can easily check that, if a polynomial $f=\sum_i f_ix^i$ satisfies $\deg(f_d)=0$, $\deg(f_{d-1})=s$ for
some $s>0$ and $\deg(f_i)< (d-i)s$ for all
$i\in[0,d-2]$, then $f$ has a Laurent root of degree $s$. This can be
generalized to up to $d$ roots.

  For any edge
 $(v_i,i)\leftrightarrow(v_j,j)$ of the Newton polygon we call length the
 quantity $|i-j|$.   
 \begin{prop}(Section II.6,\cite{Neu99})\label{cor_slope} If $f\in\F_q[t][x]$ is
	 a polynomial, each edge of length $1$ in the Newton polygon
	 corresponds to an infinite Laurent root for $f$. 
\end{prop}
For example, the polynomial $f=x^7+t^3x^6+t^5x^5+(t^6+1)x^4+t^6x^3+(t^5+t+1)x^2+t^3x+1$
has a Newton polygon with $7$ edges of length $1$ and therefore $7$ infinite Laurent roots.

Note however that the converse of the Proposition~\ref{cor_slope}
is false in general: a polynomial might have infinite Laurent roots which cannot be
counted using the Newton polygon, e.g. for the polynomial $f$ above,
$f^{1}:=f(x+t^4)$ also has $7$ infinite Laurent roots, although its Newton
polygon has no edge of length $1$.

%%}}}
%%}}}

\section{The combined effect of the properties which influence the sieving
efficiency}\label{sec:combining}

The previous sections identified three elements which affect the sieve and
defined associated measures: $\sigma$ for the size
property, $\alpha$ for the root property and $\alpha_\infty$ for the
cancellation property. The list might be extended with other properties and the
quantifying functions can be combined in different fashions in order to compare
arbitrary polynomials. In this section we define two general purpose
functions based on the three properties above and show their relevance through
experimentation. 

\subsection{Adapting Murphy's $\mathbb{E}$ to the FFS}

As a first function which compares arbitrary polynomials, we adapt Murphy's
$\mathbb{E}$, already used for the NFS algorithm (Equation $5.7$, \cite{Mur99}). Heuristically, $\mathbb{E}$
uses Dickman's $\rho$ to approximate the
number of relations found by $f$ and $g$ on a sieving domain.
\begin{deff}
Let $f,g\in \F_q[t][x]$ be two irreducible polynomials, $s$ an integer
called skewness parameter, $e$ a half integer called sieve parameter and $\beta$
an integer called smoothness bound. Let $D(s,e)$ be the set of coprime pairs
 $(a,b)\in \F_q[t]^2$ such that $0\leq \deg(a)\leq \lfloor
 e+\frac{s}{2}\rfloor$ and $0\leq
\deg(b)\leq \lfloor e-\frac{s}{2} \rfloor$. We define:
\begin{equation*}
	\mathbb{E}(f,g,s,e,\beta)=\sum_{(a,b)\in D(s,e)}
	\rho\left(\frac{\deg F(a,b)+\alpha(f)}{\beta}\right)\cdot
	\rho\left(\frac{\deg G(a,b)+\alpha(g)}{\beta}\right).
\end{equation*}

\end{deff}

Unlike the situation in the NFS case, where $\mathbb{E}$ must be approximated by
numerical methods, in the case of the FFS one can compute $\mathbb{E}$ in
polynomial time with respect to $\deg(f)$, $\deg(g)$ and $e+|s|$. Note that $\rho$ can
be evaluated in polynomial time to any precision on the interval which is
relevant in this formula. 

We recall that we focus in this work on the case where the
polynomial $g$ is linear, as in~\cite{JoLe02}. More precisely, we assume
that $g=g_1x+g_0$ is chosen with $g_0$ of much higher
degree than $g_1$. In this case, the algorithm goes as follows:
\begin{enumerate}
	\item[1.] compute the Laurent roots of $f$ up to
		$\lfloor d(e+\frac{s}{2})+\deg_t f\rfloor$ terms;
	\item[2.] for each $d_a\leq \lfloor e+\frac{s}{2}\rfloor$,
		$d_b\leq \lfloor e-\frac{s}{2} \rfloor$ and $i\in\N$, use
		Lemma~\ref{prob_r} to compute the number $n(d_a,d_b,i)$ of pairs
$(a,b)$ such that $\deg(a)=d_a$, $\deg(b)=d_b$ and $\deg(F(a,b))=i$;
\item[3.] compute 
	\begin{equation}
	 \sum_{d_a,d_b,i}n(d_a,d_b,i)\rho\left(\frac{i+\alpha(f)}{\beta}\right)
	 \cdot\rho\left(\frac{d_b+\deg
	 g_0+\alpha(g)}{\beta}\right).
	\end{equation}
\end{enumerate}

One can use Murphy's $\mathbb{E}$ to choose the optimal skewness
corresponding to a pair $(f,g)$ of polynomials. 
\begin{ex}
Consider for instance the two
polynomials used for the computation of the discrete logarithm in
$GF(2^{619})$: $f=x^6+(t^2+t+1)x^5+(t^2+t)x$+0x152a and
$g=x-t^{104}-0\text{x6dbb}$ written in hexadecimal\footnote{Each polynomial
	$\ell$ of $\F_2[t]$, $\ell=\sum_i \ell_i t^i$ with
	$\ell_i\in\{0,1\}$, is represented by base-$16$ notation of the integer $\sum_i
\ell_i 2^i$.} notation \cite{BBDGJTVZ12}. They used the smoothness bound $22$ and most of the
computations were done using special-Q's $(q,r)$ with $\deg q=25$. The pairs $(a,b)$
considered for each special-Q were $(ia_0+j a_1,ib_0+jb_1)$ with $(a_0,b_0)$
and $(a_1,b_1)$ two pairs on the special-Q lattice and $i,j$ were
polynomials of degree at most $12$. Hence, the pairs $(a,b)$ considered were
such that $\deg a+\deg b= \deg q+24$, so $e=(\deg a+\deg b)/2=24.5$. Note that, if $a$ and
$b$ have maximal degree on our set, the difference $\deg a-\deg b$ cannot be
even. Table~\ref{choose_skewness} shows that the best skewness value~is~$3$.
\end{ex}

Note though that in \cite{BBDGJTVZ12} one started by experimentally choosing
the best skewness for polynomials of a given bound on the degrees. Then they
selected polynomials which, for a given value of $s$, minimize the value of 
epsilon, the function that we define below.

\begin{table}
	\caption{Choosing the best skewness using $\mathbb{E}$. The
	parameters are set to $e=24.5$ and $\beta=22$.}
	\begin{center}
	\begin{tabular}{|c|ccccc|}
		\hline
		$s$ & $-1$ & $1$ & $3$ & $5$ & $7$\\
		\hline
		$10^{-5}E(f,g,s,e,\beta) $ & $2.54$ & $3.31$ & $3.46$
						 & $2.88$ & $2.12$ \\
		\hline
	\end{tabular}
\end{center}
\label{choose_skewness}
\end{table}

\paragraph{An alternative to $\mathbb{E}$: Epsilon}
Recall that $\sigma$ is the degree of the norm when no cancellation occurs,
$\alpha$ is the degree gained due to the modular roots and
$\alpha_\infty$ is the degree gained thanks to cancellations. It seems
natural that their sum is the degree of a polynomial which has the same
skewness probability as $F(a,b)$ for an ``average'' pair $(a,b)$ on the
sieving domain.
\begin{deff} For a polynomial $f\in\F_q[t][x]$, a skewness
	parameter $s$ and a sieve parameter $e$, we call epsilon the
	following average degree
\begin{equation*} 
\epsilon(f,s,e)=\alpha(f)+\alpha_\infty(f,s)+\sigma(f,s,e).
\end{equation*}
\end{deff}

Epsilon can be used to estimate the speedup of a polynomial with good
properties. For example if the smoothness bound is $28$ and two polynomials
have the value of epsilon equal to $107$ and $109$ respectively, then we expect
a speedup of $\rho(107/28)/\rho(109/28)\approx 1.19$.

\paragraph{Comparing $\epsilon$ and $\mathbb{E}$}
Since the subroutines necessary in the computation of $\epsilon$ are equally
used when evaluating $\mathbb{E}$, in practice epsilon is faster to compute than~$\mathbb{E}$. The advantage of $\mathbb{E}$ is that it is more
precise, but the experiments of the next section will show that epsilon is
reliable enough.

\subsection{Experimental validation}

%%{{{

The implementation used in the experiments is the one described
in~\cite{DeGaVi12}, which is freely available at~\cite{cadonfs}. 
To our knowledge, no other press-a-button 
implementation of the FFS is publicly available. In addition, this
implementation does relatively few modifications which could loose relations, making a theoretical study
inexact. 

The real-life efficiency of a polynomial is measured either by the
number of relations per second, by the total number of relations, or as the
average number of relations per special-Q. We kept the last one as a
measure of efficiency since the software 
offers an option to reliably measure it and because it considers only the
polynomial properties rather than the implementation quality.

\begin{figure}
	\begin{center}
	\includegraphics[width=10cm]{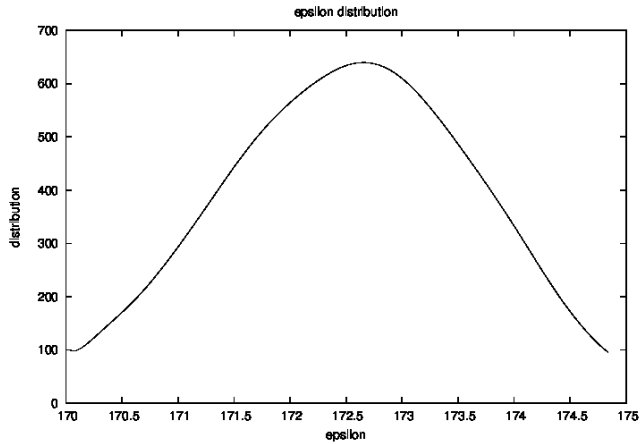}
	\end{center}
	\caption{Distribution of epsilon on a sample of $20000$ polynomials of
	$\F_2[t][x]$ of degree $6$ in $x$ and $12$ in $t$.}
	\label{fig:epsilon}
\end{figure}

\begin{experiment}
	\label{valid}

	We selected a sample of polynomials $f$ after evaluating epsilon for a
	range of polynomials considered one after another in lexicographical
	order starting from $x^6+(t^2+t+1)x^5+(t^3+t^2+t+1)x^3+tx+t^{11}$.
	Note that the choice of the starting point guarantees that the
	polynomials considered have at least one infinite Laurent root.
	Since the distribution of epsilon was that of Figure~\ref{fig:epsilon}, most of the polynomials tested had values
	of epsilon in a narrow interval. This lead us to
	select only one polynomial in each interval of length $0.01$, to a
	total of $119$ polynomials. Next we extended the sample with $60$
        polynomials starting from $x^6 + t^3x^5 + (t^5 + 1)x^4 + t^6x^3 +
	t^6 + 1$. For each polynomial $f$ we associated a random monic linear
	polynomial $g$ suited to the FFS, having degree in $t$ equal to
	$104$. Indeed, as shown in Theorem~\ref{alpha_x} and in section~\ref{sssec:alpha_infty} respectively, linear polynomial have the same values of
	$\alpha$ and $\alpha_\infty$ respectively.

	We set the parameters as follows: 
        $\text{I}=\text{J}=12$,
	$\text{fbb}0=\text{fbb}1=22$,
$\text{lpb}0=\text{lpb}1=28$, $\text{thresh}0=\text{thresh}1=100$,
$\text{sqside}=1$. The polynomials $q$ used in the special-Q technique
were the first irreducible ones starting from $t^{25}$. We called the
option ``reliablenrels'' which tests as many values of $q$ as needed in order to
obtain a measurement error of $\pm3\%$ with a confidence level of $95\%$. The
skewness parameter was set to $S=3$ because, for the finite fields where the
degree-$6$ polynomials are optimal, this is a sensible choice. Finally,
the parameter sqt was set to $1$ so that, for most special-Q's, the
sieving domain was such that $\deg a\leq 26$ and $\deg b\leq 23$.

\begin{figure} 	
\centering
	\caption{Epsilon and sieve efficiency for the
	polynomials $f$ in Experiment~\ref{valid}. The function $h$ is a
function of type $a+bx+c\log x$, with no special significance. }
\includegraphics[width=12cm]{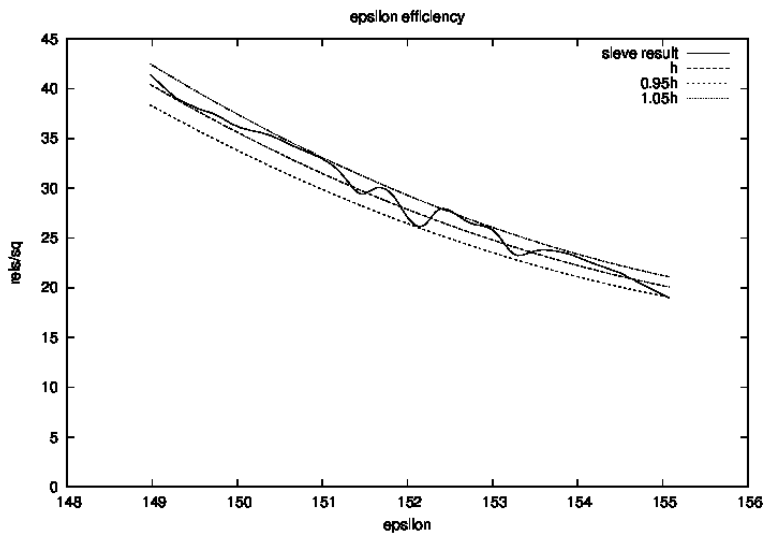}
\label{fig-validation}
\end{figure}

\end{experiment}

The results plotted in Figure~\ref{fig-validation} indicate that the
sieve efficiency is not far from a strictly decreasing function of
epsilon. To illustrates this, we plotted a decreasing function that
fits our results, such that the relative error of our measurements is
always less than $5\%$.

Finally, one can see that a sensible choice of the polynomial can save a
factor $2$ in the sieve time when compared to a bad choice.

%%}}}
\subsection{Correlation between $f$ and $g$}
\label{sec:correlation}
%%{{{
A standard heuristic states that the probabilities of $F(a,b)$ and $G(a,b)$
to be smooth are independent, e.g. Murphy's $\mathbb{E}$ multiplies the two
probabilities. The inexactness of this approximation could be called
correlation property. According to Experiment~\ref{valid}, the correlation property
has a small effect on the sieve, so that we bound ourselves to illustrate it
by an example and a practical experiment. 

\begin{ex}\label{correlation} Let $f=x^2-t^2(t+1)$, $g_1=x-(t^2+t)^7$ and
$g_2=x-(t^2+t+1)^7$. Let $F$,
$G_1$ and $G_2$ be the homogenizations of $f$, $g_1$ and $g_2$
respectively. Note that, for coprime pairs $(a,b)\in \F_q[t]$, $F(a,b)$ is divisible
by $t$ if and only if $a\equiv 0\mod t$. For these pairs 
we have $G_1(a,b)\equiv 0\mod t$ whereas $G_2(a,b)\not\equiv 0\mod t$.
In short, $g_1$ increases the number of doubly smooth pairs whereas $g_2$
that of pairs which are smooth on the rational or the algebraic side, but not on both. 
\end{ex}
If $\deg_xg=1$, then every prime
power $\ell^k$ of $\Res(f,g)$ implies a correlation between the events
$\ell^k \mid F(a,b)$ and $\ell^k \mid G(a,b)$ on a domain of pairs $(a,b)$. 
 Table~\ref{fg} summarizes an
experiment in which we compared different pairs $(f,g)$ with $f$ having
similar values of $\epsilon$. We selected three polynomials $f$ of the form
$f=\tilde{f}+f_0$ with $\tilde{f}=x^6+tx^5+(t+1)x^4+(t^2+t+1)x^3$ and we
associated to each one a linear polynomial $g$ of the form
$g=\tilde{g}+g_{00}$ with $\tilde{g}=x-t^{104}-t^{14} + t^{13} + t^{11} +
t^{10} + t^8$. Instead of imposing that $\Res(f,g)$ has an irreducible factor of degree
$619$ as in the previous experiment, we aimed to find polynomials $g$ such
that $\Res(f,g)$ has first no, then many, small factors. The experiment
indicated that the correlation property explains part of the error observed
in Experiment~\ref{valid}, bringing it close to $3\%$, which is equal to our measurement error. 
\begin{table}
	\caption{Influence of the linear polynomials $g$ on the sieve
	efficiency. $f=\tilde{f}+f_0$ and $g=\tilde{g}+g_{00}$ with
	$\tilde{f}$ and $\tilde{g}$ given in section~\ref{sec:correlation}.}
	\begin{center}
\begin{tabular}{|c|c|c|c|}
\hline
$f_0$ & $g_{00}$ &small factors of $\Res(f,g)$&
efficiency (rels/sq)\\
\hline
$0\text{x}19$ & $0\text{xb2}$ & $-$            &$3.3$\\
$0\text{x}12$ & $0\text{xbf}$ & $-$            &$3.4$\\
$0\text{x}12$ & $0\text{xb8}$ & $-$            &$3.4$\\
\hline
$0\text{x}12$ & $0\text{xae}$ &$t\cdot (t+1)$          &$3.5$\\
$0\text{x}12$ & $0\text{xa0}$ &$t\cdot (t+1)$      &$3.5$\\
\hline
$0\text{x}12$ & $0\text{xbb}$ &$t\cdot (t+1)\cdot (t^3+t^2+1)$        &$3.5$\\
$0\text{x}1e$ & $0\text{xbb}$ &$t\cdot (t+1)\cdot(t^3+t+1)\cdot (t^6+\cdots)$     &$3.6$\\
\hline
\end{tabular}
\end{center}

\label{fg}
\end{table}

Since it is easy to associate many linear polynomials $g$ to a unique $f$ and
since linear polynomials have the same value of epsilon, it is interesting
to select $g$ such that $\Res(f,g)$ has many small factors. Nevertheless,
$f$ and $g$ are chosen such that $\Res(f,g)$ has degree $d
(\lfloor \frac{n}{d}\rfloor+1)$ and we require an irreducible factor of
degree $n$, leaving little room for additional factors. Moreover, when
$f$ imposes an extra factor to $\Res(f,g)$ (for example by having
$1$ projective and $q$ affine roots modulo $t$), depending on the
congruence of $n$ modulo $d$, it can be impossible to choose $g$ of optimal
degree in $t$. See \ref{ssec:classical} for an example.
%%}}}
\subsection{A sieve algorithm for alpha}
%%{{{
After we showed the relevance of epsilon, the polynomial selection comes to
evaluating epsilon on a large set of polynomials. One can try various ranges
of polynomials $f=\sum f_i x^i$, given by some degree bounds on their
coefficients $f_i$, which optimize sigma and/or impose a number of Laurent
roots, as shown in \ref{many}. The most time-consuming part of the
computations, the evaluation of alpha, can be done on each range by a
sieving procedure.  

The idea is that, for each
 irreducible polynomial $\ell\in\F_q[t]$, we compute $\alpha_\ell$ for all the
 residue polynomials $f$ of $\F_q[t][x]$ modulo $\ell$ and then we
 update the values of $\alpha_\ell$ for all the polynomials $f$ in the
 range. 

 Let $d$, $e_0,\ldots,e_{d-1}$ and $e_d$ be integers. We consider the range
 of the polynomials $f=\sum_{i=0}^d f_i x^i\in\F_q[t][x]$ such that for
 $i\in[0,d]$, $\deg_t f_i\leq e_i$. Call $H$ the set of values taken by the
 tuple $(f_d,\ldots,f_2)$ and $T$ those taken by $(f_1,f_0)$. Let $L$ be the
 set of irreducible polynomials up to a given bound. Let $k_\text{max}$ be a
 parameter and let us suppose
that, for all $\ell\in L$, all the roots $r\mod \ell^k$ with
$\deg(\ell^k)\geq k_\text{max}$ extend indefinitely.

For an irreducible polynomial $\ell$, Algorithm~\ref{alpha_sieve} below
computes $\alpha_\ell(f)$ for all $f$ in the range
and can be a subroutine to computing $\alpha(f)$ for the same range. We
denote by $residues(\ell^k)$ the set of polynomials in $\F_q[t]$ of degree
at most $k\deg\ell -1$.
\begin{algorithm}
\caption{The alpha sieve}
\begin{algorithmic}[1]
\State Initialize $\alpha_\ell$ to a vector of value $\deg(
\ell)/(\Norm(\ell)-1)$
\State{ $k_0 \gets \lceil k_\text{max} /\deg \ell \rceil $ }
\For{$(f_d,f_{d-1},\ldots,f_2)$ in $H$}
\For{$k$ in $[1..k_0]$ and $r$  in residues($\ell^k$)}
\For{$(f_1,f_0)$ in $T$ such that $f_1r+f_0\equiv -\sum_{i=2}^d f_i r^i\mod
\ell^k$}
\State{$f\gets\sum_{i=0}^d f_i x^i$}
\If{$k<k_0$}
\State{ $\alpha_\ell(f)$ $\gets$
$\alpha_\ell(f)$ $-$ $\deg(\ell)\Norm(\ell)^{1-k}/(\Norm(\ell)+1)$ }
\Else
\State{ $\alpha_\ell(f)$ $\gets$
$\alpha_\ell(f)$ $-$ $\deg(\ell)\Norm(\ell)^{2-k}/(\Norm(\ell)^2-1)$ }
\EndIf
\EndFor
\EndFor
\EndFor
\end{algorithmic}
\label{alpha_sieve}
\end{algorithm}

The correctness of Algorithm~\ref{alpha_sieve} follows from
Proposition~\ref{ca}. For a fixed value of $k_\mathrm{max}$,
the complexity per polynomial is
$O(1)$, as the most time-consuming steps are those in
lines $8$ and $10$. For comparison, in the naive algorithm, for each
polynomial, one needs to find the roots modulo $\ell$, which takes a
non-constant polynomial time in $d+\deg(\ell)$. In practice,
Algorithm~\ref{alpha_sieve} showed to be much faster, as Paul Zimmermann
used it to compute $\alpha_\ell(f)$, for all the irreducible polynomials
$\ell$ with $\deg\ell\leq 6$, on the range of the $2^{48}$ monic polynomials $f\in\F_2[t][x]$ of degree $6$ such that for
$i\in[0,6]$, $\deg_t f_i\leq 12-2i$.
%%}}}
%%{{{

\section{Sieving with inseparable polynomials}\label{insep}
\subsection{Particularities of the inseparable polynomials}
%%{{{
Despite the possibility of adding new technicalities, the inseparable
polynomials have been preferred in two record computations
\cite{Japanese2010}, \cite{Japanese2012}. Moreover, the Coppersmith
algorithm, implemented in \cite{Tho03}, can be seen as a
particular case of the FFS, using inseparable polynomials. In order to
present the Coppersmith algorithm from this point of view and in order to compare inseparable polynomials to separable ones, we start
with their definition, followed by their main properties.

\begin{deff}
	An irreducible non-constant polynomial $f\in\F_q[t][x]$ is said inseparable if
	$f'=0$, where $f'$ denotes the derivative with respect to $x$.
\end{deff}

	For every inseparable polynomial $f$, there exists a power of the
	characteristic of $\F_q$, $d$, and a polynomial
	$\hat{f}\in\F_q[t][x]$ such that $f=\hat{f}(x^d)$ and $\hat{f}'\neq
	0$. This simple property allows us to factor any irreducible
	polynomial $\ell$ in the function field of $f$ in two steps. First
	we factor $\ell$ in the function field of $\hat{f}$, then we further
	factor every prime ideal $\mathfrak{l}$ of $\hat{f}$. The main
	advantage is that some prime ideal factorization algorithms work
	only for separable polynomials (for example
	Magma implements the function fields only in the case of separable polynomials 
	\cite{MR1484478}). The factorization of the ideals $\mathfrak{l}$ of
	$\hat{f}$ in the function field of $f$ is easy using the following
	result.

	\begin{prop}(Corollary X.1.8,\cite{Lorenzini})\label{prop:lL} Let $p>0$ be a prime and $q$
	and $d$ two powers of $p$. Let $\hat{K}/\F_q(t)$ be a function field. Let
	$K/\hat{K}$ be an extension of polynomial $x^d-\theta_1$ with
	$\theta_1\in \hat{K}$. Then
	every prime ideal $\mathfrak{l}$ of $\hat{K}$ decomposes as
	\begin{equation}\label{lL}
	\mathfrak{l}\mathcal{O}_K=\mathfrak{L}^d
\end{equation}
for some prime ideal $\mathfrak{L}$ such that $\mathfrak{L}\bigcap
\mathcal{O}_{\hat{K}}=\mathfrak{l}$.
\end{prop}

In the FFS algorithm, it is required to compute for each smooth element
$a-b\theta$ of the function field of $f$, the valuation of every prime ideal
$\mathfrak{L}$ in the factor base. For this, we start by factoring
$(a-b\theta)^d$ in the integer ring of the function field $\hat{K}$ of
$\hat{f}$:
$$(a-b\theta)^d\OO_{\hat{K}}=(a^d-b^d\theta_1)\mathcal{O}_{\hat{K}}=\prod_i
\mathfrak{l}_i^{e_i}$$ and then we obtain $(a-b\theta)\mathcal{O}_K=\prod_i
\mathcal{L}_i^{e_i}$ where the $\mathfrak{L}_i$ are such that
$\mathfrak{l}_i\mathcal{O}_K=\mathfrak{L}_i^d$. 
%%}}}
\subsection{Speed-up in the FFS due to the inseparability}
\label{ssec:speedup}
%%{{{
\begin{deff}
	Let $f$ and $g$ be two polynomials of $\F_q[t][x]$ such that
	$\Res(f,g)$ has an irreducible factor of degree $n$. Assume that
	$\deg_xg=1$ and write $f=\hat{f}(x^d)$ for some separable polynomial
	$\hat{f}$ and some integer $d$ which is either $1$ or a power of
	char($\F_q$). We call free relation any irreducible polynomial $\ell\in\F_q[t]$ such that $\ell\nmid
	\Disc(\hat{f})f_d$ and $(f\mod \ell\F_q[t][x])$ splits into degree-$1$ factors. 
\end{deff}
Clearly each free relation of norm less than the smoothness bound creates an
additive equation between the virtual logarithms of the ideals in the factor
base.

The number of free relations is given by Chebotarev's Theorem as follows.
First note that, due to Proposition~\ref{prop:lL}, a polynomial $\ell$ is a free relation for $f$ if and only
if it is a free relation for $\hat{f}$. Then, the proportion of free
relations among the irreducible polynomials is, according to Chebotarev's Theorem, asymptotically equal to the
inverse of the cardinality of the Galois group of the splitting field of
$\hat{f}$. Call $N$ the number of monic
irreducible polynomials in $\F_q[t]$ of
degree less than the smoothness bound.
Then, the number of free relation is:
\begin{equation}
	\# \{\text{free relations}\}=\frac{N}{\#\Gal(\hat{f})}.
\end{equation}

We compare this to the cardinality of the factor base. Since the cardinality
of the rational side is $N$ and $f$ has as many ideals as
$\hat{f}$, it is enough to evaluate the cardinality of the algebraic side. According to
Chebotarev's Theorem, the number of pairs $(\ell,r)$ such that
$f(r)\equiv0\mod \ell$, $\deg r <\deg \ell$ and $\deg \ell$ is less than the
smoothness bound is $\chi N$ where $\chi$ is the average number of
roots of $\hat{f}$ fixed by the automorphisms of the splitting field of
$\hat{f}$. It can be checked that each root of $\hat{f}$ is fixed by a
fraction 
$1/\deg(\hat{f})$ of the automorphisms, so $\chi=1$. Hence, asymptotically the
factor base has $2N+o(N)$ elements. 

Heuristically, $\Gal \hat{f}$ is the full symmetric group for all but a
negligible set of polynomials $\hat{f}$, so most often we have
$\#\Gal(\hat{f})=\deg(\hat{f})!$. We list the
results for $\deg f$ equal to $6$ and $8$ in Table~\ref{freerelations}. The case in which $d=3$ and $\deg \hat{f}=2$ brought a $\frac{4}{3}$-fold
speedup in \cite{Japanese2010}.
\paragraph{Coppersmith algorithm}
The case $d=8$ and $\deg \hat{f}=1$
corresponds to the Coppersmith algorithm.
Indeed, since half of the relations are free relations, the sieve is
accelerated by a factor of $2$. Moreover, since $\deg(\hat{f})=1$, the free
relations are particularly simple, linking exactly one element in the
rational side to one element in the algebraic side of the factor base (Proposition~\ref{prop:lL}).
Therefore one can rewrite the relations using only the elements in the
rational side, hence speeding up the linear algebra step by a factor of $4$.

\begin{table}

	\caption{Number of free relations of a pair $f,g$ with
		$f=\hat{f}(x^d)$,
		$\hat{f}$ separable and $\deg(g)=1$. $N$ is the number of
		irreducible monic polynomials of degree below the smoothness bound. The computations assume
that $\#\Gal(\hat{f})=(\deg \hat{f})!$.}
\begin{center}
\begin{tabular}{|c|c|c|c|c|}
\hline
char($\F_q$)  & $d$ & $\deg(\hat{f})$ & $\#\{$factor base$\}$& $\#\{\text{free relations}\}$\\
\hline
any                            & $1$   &  $6$     & $2N$      & $N/720$    \\
\hline                                          
        $2$                     & $2$   &  $3$    &  $2N$      & $N/6$      \\
        \hline                                           
        $3$                     & $3$   &  $2$    &  $2N$      & $N/2$      \\  
           \hline                                        
	   \hline                                        
        any                     & $1$   &  $8$    &  $2N$     & $N/40320$   \\
        \hline                                       
        $2$                     & $2$   &  $4$    &  $2N$     & $N/24$      \\
        \hline                                           
	$2$                     & $4$   &  $2$    &  $2N$     & $N/2$       \\
        \hline                                         
	$2$                     & $8$   &  $1$    &  $2N$     & $N$         \\
\hline

\end{tabular}
\end{center}

	\label{freerelations}
\end{table}
%%}}}
\subsection{Root property of inseparable polynomials}
\label{inseparable}
Despite the fact that the inseparable polynomials are relatively few, being
possible to exhaustively test them, it has its own interest to understand why
inseparable polynomials have a bad root property and in particular, why the
alpha value of many polynomials used in the Coppersmith algorithm is $2$.
Note that our proof that alpha converges covers only the case of separable
polynomials. In this section we give some results on their root property.  

First, the number of pairs $(\ell,r)$ with $\ell$ irreducible and $r$ a
polynomial
of degree less than $\deg(\ell)$ such that $f(r)\equiv 0\mod \ell$ has a
narrower range of values than it does for the separable polynomials. Indeed, as shown by the following result, this number
corresponds to the number of roots of $\hat{f}$. The bounds in Theorem~\ref{alpha_convergence}, when written explicitly, are narrower for polynomials of
degree $\deg(\hat{f})$ than for those of degree $\deg f$. For example, if
$\hat{f}$ is linear, this number is a constant.
\begin{lem} \label{lem:hat}Let $\hat{f}\in\F_q[t][x]$ be a polynomial, $d$ a power of the
	characteristic of $\F_q$ and $f=\hat{f}(x^d)$. Let $\ell$ be an
	irreducible polynomial in $\F_q[t]$. Then there is a
bijection between the sets $\{\hat{r}\in\F_q[t]\mid \deg
\hat{r}<\deg \ell, \hat{f}(\hat{r})\equiv 0\mod \ell\}$ and
	$\{r\in\F_q[t]\mid \deg r< \deg\ell, f(r)\equiv 0\mod \ell\}$. 
\end{lem}
\begin{proof}
	The non-null residues of $\ell$ form a group of cardinality
	$\Norm(\ell)-1$, which is coprime to $q$ and hence to $d$. Therefore
any root $\hat{r}$ accepts one and only one $d^\text{th}$ root modulo
	$\ell$. 
\end{proof}
The second reason for having bad values of alpha is that most of the roots
modulo irreducible polynomials $\ell$ do not lift to roots modulo $\ell^2$.
Recall the following classical result.
\begin{lem} Let $\ell\in\F_q[t]$ be an irreducible polynomial. Write
	$(\F_q[t]/\langle\ell^2\rangle)^*$ for the group of residues modulo
	$\ell^2$ which are not divisible by $\ell$. Put
	$U=\{e^{\Norm(\ell)}\mid e \in (\F_q[t]/\langle\ell^2\rangle)^*\}$
	and $V=\{1+\ell w\mid \deg w<\deg
	\ell\}$. Then we have 
	\begin{equation*}
		(\F_q[t]/\langle\ell^2\rangle)^*\simeq U\times V.
	\end{equation*}
\end{lem}
The group $U$ has order $\Norm(\ell)-1$ which is coprime to $d$, so
$d^\text{th}$ roots always exist and are unique in $U$. On the other hand, in
$V$, only the neutral element is a $d$-th power.
As a consequence only a fraction $1/\#V=1/\Norm(\ell)$ of the residues $\hat{r}$
modulo $\ell^2$ can have $d^\text{th}$ roots modulo $\ell^2$. Let us make
the heuristic that the roots of $\hat{f}$ modulo $\ell^2$ are random elements
of $\F_q[t]/\langle \ell^2\rangle$. Then only a small fraction of the roots
of $f$ lift modulo squares of irreducible polynomials and, for a non
negligible fraction of polynomials $f$ no root $r$ modulo some irreducible
polynomial $\ell$ lifts modulo $\ell^2$.  

Among the Coppersmith polynomials $f$, i.e. such that $\hat{f}$ is
linear, many $f$ are such that no modular root of $f$ lifts modulo squares.
Let us compute the value of alpha in this situation.
\begin{lem}
	Let $\hat{f}$ be a linear polynomial of $\F_q[t][x]$, $d$ a power
	of the characteristic of $\F_q$ and put $f=\hat{f}(x^d)$. Assume
	that there is no pair of polynomials $\ell$ and $r$ with $\ell$
	irreducible and $r$ of degree less than that of
	$\ell^2$ such that $f(r)\equiv0\mod\ell^2$. Then we have
	$$\alpha(f)=\frac{2}{q-1}.$$
\end{lem}
\begin{proof}
	By Lemma~\ref{lem:hat}, for all $\ell$, $f$ has exactly $n_\ell=1$ affine
	or projective roots modulo $\ell$. 
	By Corollary~\ref{cor_ell}, for all irreducible polynomial
	$\ell$ we have
	\begin{equation*}
		\alpha_\ell(f)=\deg\ell\bigg(\frac{1}{\Norm(\ell)-1}-\frac{1}{\Norm(\ell)}\frac{\Norm(\ell)}{\Norm(\ell)+1}\bigg)=2\frac{\deg\ell}{q^{2\deg\ell}-1}.
\end{equation*}
Hence we obtain $\alpha(f)=2 \bigg(\sum_{k\geq 1}\frac{kI_k}{q^{2k}-1}\bigg)$ with $I_k$
the number of irreducible monic polynomials of degree $k$ in $\F_q[t]$. The
sum in the parenthesis was computed in Proposition~\ref{alpha_x} and equals
$1/(q-1)$. This completes the calculations.

\end{proof}
\section{Applications to some examples in the literature}\label{sec:examples}

\subsection{Thom\'e's record using the Coppersmith algorithm}
Thomé \cite{Tho03} solved the discrete logarithm problem in $\F_{2^{607}}$
using the Coppersmith algorithm. Following this algorithm, one sets
$g=x-t^{152}$ and $f=x^4+t\lambda$ for some polynomial $\lambda\in\F_q[t]$ such that
$t^{607}+\lambda$ is irreducible. The polynomial
$\lambda_0=t^9+t^7+t^6+t^3+t+1$ used by Thomé minimizes the degree of
$\lambda$. If one searches for an alternative, it is neccessary to increase
$\deg_t f$, but this is possible without affecting much the size property.
Indeed, the sensible choice is to set the skewness $s$ to $7$, so sigma does not
vary much if one increases $\deg f_0$. By testing the polynomials $\lambda$
with $\deg \lambda \leq 18$, we determined that the best alpha corresponds to
$f_1=x^4+t(t^{16} + t^{12} + t^{11} + t^7 + t^4 + 1)$. We compare the two
polynomials in Table~\ref{table} using the functions defined in this article
as well as the sieve efficiency measured with the implementation
of~\cite{DeGaVi12} and the parameters
in Experiment~\ref{valid}.
\begin{table}
	\begin{center}
	\begin{tabular}{|c|rccccc|}
	\hline	               & $\alpha(f)$ & $\alpha_\infty(f,s)$ &
		  $\sigma(f,s,e)$ & $\epsilon(f,s,e)$ &
		       $E(f,g,s,e,\beta)$& efficiency\\
		\hline
		$f_0$ &$1.27$  & $0$ & $108.12$   &$109.39$ & $1.82\cdot10^8$   & $15.2$\\
	 $f_1$ &$-1.05$ & $0$ & $108.42$   &$107.36$ & $2.10\cdot 10^8$ & $18.8$\\
		\hline
	\end{tabular}
\end{center}
\caption{Coppermith polynomials for $\F_{2^{607}}$. The parameters in the
	table are
$s=7$, $e=24.5$ and $\beta=28$. The efficiency, measured in rels/sq, uses the
parameters in Experiment~\ref{valid}.}
\label{table}
\end{table}

\subsection{Joux-Lercier's implementation of the classical variant of the FFS}
\label{ssec:classical}
Joux and Lercier \cite{JoLe02,JoLe07} considered the fields $\F_{2^n}$ with
$n=521$, $607$ and $613$. For $n=607$ they set $f_2=x^5 + x + t^2 + 1$ and
$g_2=(t^{121} + t^8 + t^7 + t^5 + t^4 + 1)x + 1$. If one searches for an
alternative, the sensible choice is to improve the root property without
changing the size property. Since $\deg_tf_2=2$ we tested all the
polynomials whose degree in $t$ is $1$ or $2$. 
\begin{experiment}
	There are $2^{18}$ polynomials $f$ such that $\deg_tf\leq 2$, out
	of which $2^{12}$ have $\deg_tf\leq 1$. 
	There were $1776$ irreducible polynomials with $\deg_tf\leq 1$ whose
	alpha is below $3$. There were $650$ irreducible 
	polynomials $f$, with $\deg_t f\leq 2$, whose alpha is negative and
	such that the partial sum of alpha up to degree $6$ is less than $0.5$. 
	
	The best $10$ values for
$\epsilon$ with skewness $s=0$ and sieve parameter $e=24.5$ were all obtained
for polynomials $f$ with $\deg_t f=2$. The best value was that of
$f_3=(t^2 + t)x^5 + (t^2 + t + 1)x^4 + (t + 1)x^3 + t^2x^2 + t^2x + t^2$. We
could not associate a linear polynomial $g$ with $\deg_tg=121$ because
$\Res(f,g)$ is always divisible by $t$ (see \ref{sec:correlation} for more
details), hence we took $g_3=x + t^{122} +
t^{13} + t^{11} + t^6 + t^5 + t^3 + t^2$. The best $f$ for which we could
select a linear polynomial $g$ of degree $121$ was
$f_4=(t^2+t+1)x^5+(t^2+t+1)x^4+x^3+(t^2+t+1)x^2+(t^2+t+1)x+t^2+t$, for
which we took $g_4=x + t^{121} + t^{12} + t^{11} + t^8 + t^6 +
t^2 + 1$. In Table~\ref{table:JL} we compare $(f_3,g_3)$ and $(f_4,g_4)$ to
$(f_2,g_2)$. Note also that all the polynomials $f$ tested have a small genus,
which could explain the small variance of alpha when $\deg_t f\leq 2$.  
\end{experiment}

\begin{table}
\begin{center}
	\begin{tabular}{|c|rccccc|}
	\hline
	$f,g$ & $\alpha(f)$ & $\alpha_\infty(f,s)$ & $\sigma(f,s,e)$  &
	$\epsilon(f,s,e)$ & $E(f,g,s,e,\beta)$ &efficiency \\
	\hline
	$f_2,g_2$ & $2.15$ &   $0$        & $122.33$ & $124.46$ &
     $8.54\cdot 10^8$ & $66.0$  \\ 
	$f_3,g_3$ & $-0.24$&   $0$        & $123.66$ & $123.36$ & $8.64\cdot
	      10^8$ & $73.8$    \\
	$f_4,g_4$ & $-0.10$&   $0$        & $123.66$ & $123.42$ & $9.49\cdot
		10^8$ & $76.0$    \\
	\hline
	\end{tabular}
\end{center}
\caption{Classical FFS polynomials for $\F_{2^{607}}$. The parameters are
$s=1$, $e=24.5$, $\beta=28$. The efficiency, measured in rels/sq, uses the
parameters in Experiment~\ref{valid}.}
\label{table:JL}
\end{table}
\subsection{Joux-Lercier's two rational side variant}

We recall briefly the ``two rational sides'' variant of~\cite{JoLe06}, 
and study its properties according to our criteria. This variant selects two
polynomials $f=\gamma_1(x)-t$ and
$g=x-\gamma_2(t)$ for some $\gamma_1$ and $\gamma_2$ in $\F_q[t]$. Then
one collects coprime pairs $(a(t),b(t))\in\F_q[t]$ such that 
both 
\begin{equation}\label{x-side-t-side}
	a(\gamma_1(x))-xb(\gamma_1(x))\text{ and }a(t)-\gamma_2(t)b(t)
\end{equation}
 are $\beta$-smooth for some smoothness bound $\beta$. 
It can be easily checked that the expressions in Equation~\ref{x-side-t-side} have
precisely the same degrees as in the classical FFS where we consider the norms of
$a-bx$ with respect to the function fields of $f$ and $g$ when $\deg_tf=1$ and $\deg_xg=1$. 
Therefore, this variant does not overpass the classical one in terms of size property. 

As for the root properties, note that the expression in Equation~\ref{x-side-t-side} which is a
polynomial in $\F_q[t]$ is the norm of $a-bx$ with respect to the linear
polynomial $g$, so its alpha value is constant and equal to one of the linear
side of the classical variant. On the other hand, the root property of the polynomial in 
$\F_q[x]$ in Equation~\ref{x-side-t-side} can not be directly measured
with our definition of alpha. Still, 
the number of
polynomials $f$ among which are selected to have a good root property is
small, so we can not expect a large deviation; e.g. there are $2^6$ values of $f$ such
that $\deg_xf=6$.

\subsection{Records on pairing-friendly curves}
The fields $\F_{3^{6n}}$ are of particular interest in cryptography
as one can break the cryptosystems which use pairing-friendly curves over
$\F_{3^n}$ by solving the discrete logarithm problem in $\F_{3^{6n}}$. The
recently proposed algorithm of Joux \cite{Jo13} proved to be very fast for
the fields of composite degree. It rendered FFS obsolete in this case and
drastically reduced the security of these curves. This section is also
interesting for illustrating the behaviour separable polynomials.

These fields allow us to run the FFS with a base field $\F_{3^d}$ with
$d=2,3$ or $6$.  Hence we collect coprime pairs $(a,b)$ of polynomials
in $\F_{3^d}[t]$ such that both $F(a,b)$ and $G(a,b)$ factor into small
degree polynomials of $\F_{3^d}[t]$. The Galois variant
\cite{JoLe06} consists in choosing the polynomials $f$ and $g$ to have their
coefficients in $\F_3[t]$ rather than $\F_{3^d}[t]$. Its main advantage is
that, due to Galois properties, the factor base is reduced by a factor
of $d$.

Hayashi \textit{et al.} \cite{Japanese2010} used the base field $\F_{3^6}$
and the polynomial $f_i=x^6+t$ to break curves over $\F_{3^{71}}$. Two years
later, Hayashi \textit{et al.} \cite{Japanese2012} used the
base field $\F_{3^3}$ and again $f_i=x^6+t$ to break cryptosystems over
$\F_{3^{97}}$. 

Since the polynomial $f_i$ is inseparable, as explained in
\ref{ssec:speedup}, one quarter of the relations collected by the FFS are free. This roughly translates
into a $4/3$-fold speedup with respect to the separable polynomials having
the same sieve efficiency. Note that $f_i$ has the best epsilon among the
$486$ inseparable polynomials $f$ in $\F_3[t][x]$ with
$\deg_tf\leq 1$ and $\deg_xf=6$. 

A better choice can be only a separable polynomial with a
better efficiency. Since the efficiency of a polynomial depends on the base
field of the factor base, we distinguish the case of $\F_3$ from the case
of $\F_{3^d}$ with $d=2,3$ or $6$.

\begin{experiment}\label{exp:char3}
	Since $\deg_t(f_i)=1$ we can use any of the $8\cdot 3^{12}$
	polynomials $f$ in $\F_3[t][x]$ such that $\deg_tf\leq 1$, without
        changing the size property.  The best alpha with respect to
        $\F_{3}$ corresponded to $f_s=tx^6 -  tx^4 + (-t + 1)x^3 + (t - 1)x + t$.
\end{experiment}
Since alpha has a small variance on the polynomials tested in
Experiment~\ref{exp:char3}, we also consider the separable polynomial
$f_{s'}=x^6 - x^2 + (t^8 + t^6 - t^4 + t^2 + 1)$, which is well suited when
the skewness parameter is set to $1$. Table~\ref{tab:char3} compares $f_s$
and $f_{s'}$ to $f_i$ for some randomly chosen
linear polynomials. 
\begin{table}
	\begin{center}
		\begin{tabular}{|c|rccccc|}
	\hline
	$f,g$ & $\alpha(f)$ & $\alpha_\infty(f,s)$ & $\sigma(f,s,e)$ &
	$\epsilon(f,s,e)$ & $E(f,g,s,e,\beta)$ & efficiency \\
	\hline
	$f_i$ & $1.33$ &   $0$        & $94.00$ & $95.33$  & $1.03\cdot10^8$  & $14.6$ \\ 
	$f_s$ & $0.29$ &   $0$        & $94.75$ & $95.04$  & $1.23\cdot10^8$  & $17.0$ \\
$f_{s'}$   & $-3.67$   &   $0$        & $96.75$ & $93.03$  & $1.61\cdot10^8$  & $21.3$ \\
	\hline
	\end{tabular}
	\end{center}
	\caption{Polynomials $f$ for fields of characteristic $3$. The last
	column was obtained using same software as in Experiment~\ref{valid}
and with parameters fbb0=fbb1=14, lpb0=lpb1=17, S=1 and $q0=t^{15}$.}
	\label{tab:char3}
\end{table}

In the case when the base field is $\F_{3^6}$ and $\F_{3^3}$ the evaluation of
alpha is slower, with a factor of $200$ compared to the case of $\F_{3^6}$. Note first
that the polynomial $f_s$, whose root property over $\F_3$ is better than
that of $f_i$, has a
poorer value of alpha when the base field is $\F_{3^6}$. In
Table~\ref{tab:normalized_alpha} we use $\frac{\log q}{\log 2}\times \alpha$ as an
alternative of alpha which allows us to compare polynomials $f$ with
coefficients in different rings $\F_q[t]$. The values of alpha are
approximated by considering only the contribution of at most $1000$ irreducible polynomials
in $\F_{3^d}$ with $d=1,2,3$ or $6$. Note that the values of $\frac{\log
q}{\log 2}\times \alpha$ are close to each other when $q=3^6$. This opens the question of
how does the distribution of alpha evolve when we compute it with respect to
a factor base in $\F_{3^6}[t]$ but for which the polynomials $f$ are in
$\F_3[t][x]$. 

\begin{table}
	\begin{center}
		\begin{tabular}{|c|cccc|}
		\hline
		$f$ &  $\overline{\alpha}(f,\F_3)$  &
                $\overline{\alpha}(f,\F_{3^2})$ &
	       $\overline{\alpha}(f, \F_{3^3})$  &
		$\overline{\alpha}(f,\F_{3^6})$ \\  
		\hline
		$f_i$ & $2.11$  & $0.35$   & $0.53$  & $0.03$  \\
		$f_s$ & $0.46$  & $0.16$   & $0.21$  & $0.08$  \\  
		\hline
	\end{tabular}
\end{center}
	\caption{The values of alpha with respect to different base field.
	The notation $\overline{\alpha}(f,\F_q)$ denotes
$\log(q)/\log(2)\cdot \alpha(f)$ with alpha corresponding to the field
$\F_q$.}
\label{tab:normalized_alpha}
\end{table}

%%%%%%%%%%%%%%%%%%%%%%%%%%%%

\section{Conclusions and open questions}\label{conclusion}
Improving on Joux and Lercier's method of polynomial selection \cite{JoLe02}, we noted
that a unique polynomial $f$ can be used to solve the discrete logarithm
problem on a range of inputs. Since the selection of $f$ can be seen as a
precomputation, we developed a series of functions which compare arbitrary
polynomials and which are much faster than directly testing the sieve
efficiency. In particular we obtained a sieving procedure for computing
alpha, the function which measures the root property and we defined a
function for measuring the cancellation property.

The case of inseparable polynomials was of particular interest as it has no equivalent notion in the NFS
world. We showed that inseparable polynomials have the advantage of a large
number of free relations, but most of the inseparable polynomials have a bad
root property. The last section applied the new functions to some records in the
literature. 

The paper also opened some questions. First, thanks to the polynomial selection proposed in
\cite{JoLe03} this discussion could be adapted to the prime field
computations by NFS. Secondly, the proof of the convergence of alpha seems
to indicate that the distribution of alpha is influenced by the genus of the
function fields. Finally, in the case of the Galois variant, it is
interesting to know how does the variance of alpha evolve when restricted to
Galois polynomials. If the variance is small enough, the sensible choice for the
Galois variant seems to be the inseparable polynomials.
\section*{Acknowledgement}
The author is indebted to his colleagues in the Catrel project, a common
study of the discrete logarithm problem in finite
fields, especially to Cyril Bouvier, Jérémie Detrey, Pierrick Gaudry, Hanza Jeljeli,
Emmanuel Thomé, Marion Videau and Paul Zimmermann.

\bibliographystyle{alpha}
\bibliography{barbulescu}

\appendix
\section{Appendix}
%%{{{

\begin{lem}\label{lem_appendix}
	Under the notations of Proposition~\ref{ca}, Equation~\ref{appendix} holds.
\end{lem}
\begin{proof}
	First, the equality below holds as we can change the summation order of
	an absolutely convergent series. \
	\begin{equation*}
\sum_{k=1}^\infty\Prob\bigg(v_\ell(F(a,b))\geq k\bigg)=\sum_{(r,k)\in
	S(f,\ell)}\Prob\bigg ((a:b)\equiv r\bmod \ell^k\bigg).
\end{equation*}
Secondly, since for all $k$, $\Prob(v_\ell(F(a,b))=k)=\Prob(\ell^k\mid
F(a,b))-\Prob(\ell^{k+1}\mid F(a,b))$ and $k\Prob(\ell^k\mid
F(a,b))\rightarrow 0$, we have 
\begin{equation*}
\sum_{k=1}^\infty k\Prob\bigg(v_\ell(F(a,b))=k\bigg)
	=\sum_{k=1}^\infty\Prob\bigg(v_\ell(F(a,b))\geq k\bigg).
\end{equation*}
Finally, let us prove that $a_\text{hom}(f)$ exists and equals $\sum_{k=1}^\infty
k\Prob\big(v_\ell (F(a,b))=k\big)$. For each $N$ and $k$ put
\begin{equation*}
a_\text{hom}(f;N,k)=\frac{\sum\{\min(v_\ell (F(a,b)),k) \mid
\gcd(a,b)\not \equiv0 \bmod \ell, \deg a,\deg b\leq N\}}{\# \{(a,b)\mid 
\gcd(a,b)\not\equiv 0\bmod \ell, \deg a,\deg b\leq N\}   }.
\end{equation*}
Call $a_\text{hom}(f;N)$ the expression above when $\min(v_\ell (F(a,b)),k)$
is replaced with $v_\ell (F(a,b))$.
On the one hand, for any $k_0\in \N$, we have $a_\text{hom}(f;N,k_0)\leq
a_\text{hom}(f;N)$, so
\begin{equation*}
	\sum_{k=1}^{k_0}k\Prob(v_\ell (F(a,b))=k)\leq
	\lim\inf_{N\rightarrow\infty}a_\text{hom}(f;N).
\end{equation*}
On the other hand, let $k_1$ be large enough
such that all the affine and projective roots modulo $\ell^{k_1}$ are simple
and put $N=k_1\deg \ell$. Since $N=k_1\deg \ell$, the proportion of pairs
$(a,b)$ of degree at most $N$ such that $(a:b)\equiv r\mod \ell^k$ for some root
$(r,k)$ equals the probability $\Prob((a:b)\equiv r\mod \ell^k)$. Hence 
$a_\text{hom}(f;N,N/\deg\ell)=\sum_{k=1}^{k_1}k\Prob(v_\ell(F(a,b))=k)$.
Now, since the roots modulo $\ell^{k_1}$ are simple there are at most $\deg
f$ of them. Also, for $a$ and $b$ of degree bounded by $N$, the norm
$F(a,b)$ has degree bounded by $N\deg f+|f|$, with $|f|$ the maximal degree of the coefficients of $f$. Hence
\begin{equation*}
	|a_\text{hom}(f;N)-a_\text{hom}(f;N,N/\deg \ell)|\leq \frac{\deg f (N \deg
	f+|f|)}{q^{N/\deg \ell}}.
\end{equation*}
This further implies 
\begin{equation*}
\lim\sup_{N\rightarrow\infty}a_\text{hom}(f;N)\leq\sum_{k\in\N}k\Prob(v_\ell( F(a,b))=k).
\end{equation*}
We conclude that $a_\text{hom}(f;N)$ converges to $\sum_{k\in\N}k\Prob(v_\ell(
F(a,b))=k)$. 
\end{proof}
%%}}}
\end{document}

%% file: skewness.1.tex
%Uncomment next line if XeTeX is used
%\def\pgfsysdriver{pgfsys-xetex.def}

%\input pgf.tex
%\input tikz.tex
\usetikzlibrary{arrows}
\baselineskip=10pt
\hsize=6.3truein
\vsize=8.7truein
\definecolor{cqcqcq}{rgb}{0.75,0.75,0.75}
\tikzpicture[line cap=round,line join=round,>=triangle 45,x=0.6152365166685244cm,y=0.613426502285878cm]
\draw [color=cqcqcq,dash pattern=on 1pt off 1pt, xstep=0.6152365166685244cm,ystep=0.613426502285878cm] (-0.98,-5.23) grid (13.64,5.37);
\clip(-0.98,-5.23) rectangle (13.64,5.37);
\draw (0,4)-- (12,4);
\draw (12,4)-- (12,-4);
\draw (12,-4)-- (0,-4);
\draw (0,-4)-- (0,4);
\draw (0,0)-- (12,0);
\draw (6,4)-- (6,-4);
\draw (3,4)-- (3,0);
\draw (12,-2)-- (6,-2);
\draw (8.76,1.91) node[anchor=north west] {0};
\draw (4.21,2.02) node[anchor=north west] {-1};
\draw (8.68,-0.76) node[anchor=north west] {1};
\draw (3,0)-- (3.04,-4);
\draw (1.93,2) node[anchor=north west] {-2};
\draw (8.72,-2.19) node[anchor=north west] {2};
\draw (4.34,-0.79) node[anchor=north west] {0};
\draw (1.38,4)-- (1.36,-4);
\draw (1.96,-0.76) node[anchor=north west] {-1};
\draw (4.33,-2.2) node[anchor=north west] {1};
\draw [->] (0,-4) -- (0.02,5);
\draw (12.48,-3.03) node[anchor=north west] {a};
\draw (0.18,4.83) node[anchor=north west] {b};
\draw (6,-2)-- (0,-2);
\draw (12,-3.06)-- (0,-3.06);
\draw (11.74,-4.48) node[anchor=north west] {1};
\draw (5.84,-3.95) node[anchor=north west] {$\frac{1}{q} $};
\draw (-0.89,0.87) node[anchor=north west] {$ \frac{1}{q} $};
\draw (-0.48,3.97) node[anchor=north west] {1};
\draw (2.79,-3.92) node[anchor=north west] {$ \frac{1}{q^2}$};
\draw (-0.91,-1.12) node[anchor=north west] {$\frac{1}{q^2} $};
\draw (-0.35,-4.35) node[anchor=north west] {0};
\draw [->] (0,-4) -- (13,-4);
\endtikzpicture
%\bye